\documentclass{acm_proc_article-sp}

\usepackage{subfigure}
\usepackage{graphicx, graphics}
\usepackage{epstopdf}
\usepackage{times,amsmath,epsfig}
\usepackage{epsf}
\usepackage{grffile}
\usepackage{amsmath, amsfonts, amssymb}
\usepackage{color}
\usepackage{paralist}
\usepackage{cite}
\usepackage{algorithmic, algorithm}
\usepackage{multirow}
\usepackage{soul}
\usepackage{verbatim}
\usepackage{flushend}

\makeatletter
\def\@copyrightspace{\relax}
\makeatother

\newcommand{\eg}{{\em e.g.,}}

\newcommand{\para}[1]{\smallskip\noindent {\bf #1}}

\definecolor{gray}{rgb}{0.5,0.5,0.5}
\long\def\delete#1{}

\newcommand{\oneUmedian}{{$Frugal$-$1U$-$Median$~}}
\newcommand{\oneUany}{{$Frugal$-$1U$~}}
\newcommand{\twoUany}{{$Frugal$-$2U$~}}

\newcommand{\qdigest}{{$q$-$digest$~}}
\newcommand{\selection}{{$Selection$~}}
\newcommand{\gk}{{$GK$~}}
\newcommand{\ninty}{{90-\%~}}
\newcommand{\TODOQUIANG}[1]{{\bf TODO}}

\newlength{\figurewidthA}
\newlength{\figurewidthB}
\newlength{\figurewidthC}
\newlength{\figurewidthD}
\newlength{\figurewidthE} 
\newlength{\figurewidthF}
\newlength{\figurewidthG}
\newlength{\figurewidthH}
\newlength{\figurewidthI}
\newlength{\figurewidthJ}
\newcount\hour \newcount\minute
\hour=\time  \divide \hour by 60
\minute=\time
\loop \ifnum \minute > 59 \advance \minute by -60 \repeat
\def\drafttime{\ifnum \hour<13 \number\hour:%
                      \ifnum \minute<10 0\fi
                      \number\minute
                      \ifnum \hour<12 \ AM\else \ PM\fi
         \else \advance \hour by -12 \number\hour:%
                      \ifnum \minute<10 0\fi
                      \number\minute \ PM\fi}

\newcommand\omt[1]{}

\omt{\newtheorem{theorem}{Theorem}[section]}
\omt{}

\newtheorem{theorem}{Theorem}
\newtheorem{myexample}{Example}[section]
\newtheorem{example}[myexample]{Example}

\newcounter{mylemma}
\newtheorem{lemma}[mylemma]{Lemma}

\begin{document}

\title{Frugal Streaming for Estimating Quantiles:One (or two) memory suffices}

\numberofauthors{3}
\author{
\alignauthor
Qiang Ma\\
       \affaddr{Rutgers University}\\
       \affaddr{Piscataway, NJ 08854, USA}\\
       \email{qma@cs.rutgers.edu}
\alignauthor
S. Muthukrishnan\\
       \affaddr{Rutgers University}\\
       \affaddr{Piscataway, NJ 08854, USA}\\
       \email{muthu@cs.rutgers.edu}
\alignauthor 
Mark Sandler\\
       \affaddr{Google Inc.}\\
       \affaddr{New York, NY 10011, USA}\\
       \email{sandler@google.com}
}

\maketitle

\begin{abstract}
Modern applications require processing streams of data for estimating statistical quantities such as quantiles with small amount of memory. In many such applications, in fact, one needs to compute such statistical quantities for each of  a large number of groups, which additionally restricts the amount of memory available for the stream for any particular group. We address this challenge and introduce {\em frugal streaming}, that is algorithms that work with tiny -- typically, sub-streaming -- amount of memory per group.

We design a frugal algorithm that  uses {\em only one} unit of memory per group to compute a quantile for each group. For stochastic streams where data items are drawn from a distribution independently, we analyze and show that the algorithm finds an approximation to the quantile rapidly and remains stably close to it. We also propose an extension of this algorithm that
uses {\em two} units of memory per group.  We show with extensive experiments with real world data from HTTP trace and Twitter that our frugal algorithms are comparable to existing streaming algorithms for estimating any quantile, but these existing algorithms use far more space per group and are unrealistic in frugal applications; further, the two memory frugal algorithm converges significantly faster than the one memory algorithm.
\end{abstract}
\section{Introduction}
\label{intro}

Modern applications require processing streams of data for estimating statistical quantities such as quantiles with small amount of memory. A typical application is in IP packet analysis systems such as  Gigascope~\cite{Cranor03gigascope} where an example of a query is to find the median packet (or flow) size for IP streams from some given IP address. Since IP addresses send millions of packets in reasonable time windows, it is prohibitive to store all packet or flow sizes and estimate the median size. Another application is in social networking sites such as Facebook or Twitter where there are rapid updates from users, and one is interested in median time between successive updates from a user.  In yet another example, search engines can model their search traffic and for each search term, want to estimate the median time between successive instances of that search.

Motivated by applications such as these,  there has been extensive work in the database community on theory and practice of approximately estimating quantiles of streams with limited memory (\eg~\cite{Arasu04approxcounts, Lin04quantilesummaries, Babcock03maintainvarandmedian, Cormodbiasedquan, Manku98approxs,Agrawal95aone-pass,AlsabtiRanka-1997, Guha_streamorder, Shrivastava04mediansandbeyong, Cormode200558, Gilbert20021287369, Greenwald01onlinequantile}).  Taken together, this body of research has generated methods for
approximating quantiles to $1+\epsilon$ approximation with space roughly $O(1/\epsilon)$ in various models of data streams.

Our work here begins with our experience that while the algorithms above are useful, in reality, they get used within GROUPBYs, that is, there are a large number of groups and each group defines a stream within which we need to compute quantiles. In example applications above, this is evident. In IP analysis, one wishes to find median packet size from {\em each} of the source IP addresses, and
therefore the number of ``groups'' is $2^{32}$ (or $2^{128}$). Similarly, in social network application, we wish to compute the median time between updates for {\em each} user, and the number of users is in $100$'s of millions for Facebook or Twitter. Likewise, the number of ``groups'' of interest to search engines is in $100$'s of millions of search terms.  Now, the bottleneck of high speed memory manifests in a different way. We can no longer allocate a lot of memory to any of the groups! In real systems such as Gigascope, low level aggregation engines keep in memory as many groups as they can and rely on higher level aggregation to aggregate partial answers from various groups, which ends up essentially forcing the higher level aggregator to work as a high speed streamer, and proves ineffective.

Motivated by this, we  introduce the new direction of {\em frugal streaming}, that is streaming algorithms that work with tiny amount of memory per group, memory that is far less than is used by typical streaming algorithms.. In fact, we will work with $1$ or $2$ memory locations per group. Our contributions are as follows.

\begin{itemize}
\item
We present two frugal streaming algorithms for estimating a quantile of a stream. One uses $1$ unit of memory for the data stream item, and the other uses $2$ units of memory.

\item
For stochastic streams, that is streams where each item is drawn independently from a distribution, we can mathematically analyze and show how our algorithms converge rapidly to the desired quantile and how they stably oscillate around the quantile as stream progresses.

\item
We evaluate our algorithms on synthetic as well as real datasets from HTTP trace and Twitter.  In all cases, our frugal streaming algorithms perform accurately and quickly. Regular streaming algorithms known previously either are highly inadequate given our memory constraints or need significantly more memory to be comparable in accuracy. Further, our frugal algorithms have an intriguing ``memoryless'' property. Say the stream abruptly changes and now represents a new distribution;
 {\em irrespective of the past}, at any given moment, our frugal algorithms move towards the
 median of the new distribution without waiting for the new streaming items to drown out the old median. We also experimentally evaluate the performance of our frugal streaming algorithms with changing streams.
\end{itemize}

In Section~\ref{sec:pre} we present definitions and notations.  We present our $1$ unit
memory  frugal streaming algorithm in Section~\ref{sec:algs}. It is analyzed for stochastic streams in Section~\ref{sec:analysis} to give insights about its speed in approaching true quantile and its stability in the long run. Section~\ref{sec:extension} gives a $2$ unit memory frugal streaming algorithm. We discuss related algorithms and present our extensive experimental study in Section~\ref{sec:algstocompare} and ~\ref{sec:evaluations}. Section~\ref{sec:conclusion} has concluding remarks.

\section{Background and Notations}
\label{sec:pre}

Suppose values in domain $D$ are integers \footnote{For domains with non-integer values, their values can be rewritten to keep desired precision and scale up altogether to integers.} distributed over $\{1, 2, $ $3, \ldots, N\}$. Given a random variable $X$ in domain $D$, denote its cumulative distribution function (CDF) as $F(x)$, and its quantile function as $Q(x)$. In other words,  $F(Q(x))=x$ if CDF is strictly monotonic. 

$h$-th $p$-quantile is $x$ such that $Pr(X<x) = F(x) = \frac{h}{p}$, for convenience we use  $\frac{h}{p}$-quantile for the $h$th $p$-quantile.

$S$ is a sampled set from $D$. Define a rank function that gives the number of items in $S$ which are smaller than $x$, $R(x) = |S'|$ where $S' = \{ s_i\in S, s_i < x\}$. So when size of $S$ grows to infinity, $F(x) = \frac{R(x)}{|S|}$.

In this paper we consider rank $p$-quantiles, so the $\frac{h}{p}$-quantile approximation returned by algorithm is considered correct even if the approximation is not in value domain $D$. For example, if $D$ is distributed over two values $1$ and $1000$ with equal probabilities. Under value $\frac{1}{2}$-quantile, an estimation at $1000$ would be considered accurate (throughout our paper, upper median is used for even sample sizes). But any value between $1$ and $1000$ can also give us good estimation in terms of ranking.

Throughout when we refer to memory use of algorithms, each memory unit has sufficient bits to store the input domain, that is, each memory unit is $\log N$ bits. This is standard in data stream literature where a method uses $f$ words, it is really $f$ words each of which has sufficient bits to store the input, or $f\log N$ bits.

\renewcommand{\algorithmicrequire}{\textbf{Input:}}
\renewcommand{\algorithmicensure}{\textbf{Output:}}
\renewcommand{\algorithmiccomment}[1]{// #1}
\begin{algorithm}[tb!]
\caption{\oneUmedian}
\label{alg:1umedian}
\begin{algorithmic}[1]
\REQUIRE Data stream $S$, 1 unit of memory $\tilde{m}$
\ENSURE $\tilde{m}$
\STATE Initialization $\tilde{m} = 0$
\FOR {\textbf{each} $s_i$ in $S$}
  \IF{$s_i > \tilde{m}$}
   \STATE $\tilde{m} = \tilde{m} + 1$;
  \ELSIF{$s_i < \tilde{m}$}
   \STATE $\tilde{m} = \tilde{m} - 1$;
  \ENDIF
\ENDFOR	
\end{algorithmic}
\end{algorithm}

\section{Frugal Streaming Algorithm}
\label{sec:algs}

We start from median estimation problem and then generalize our algorithms to estimate any quantile of $S$.

\subsection{1 Unit Memory Algorithm to Estimate Median}
\label{subsec:1mem}

Our algorithm maintains only one unit of memory  $\tilde{m}$ which contains its estimate for the stream median, $m_S$. When a new stream item $s_i$ arrives, consider what our algorithm can do? Since it has no memory of the past beyond  $\tilde{m}$, it can do very little. The algorithm ``drifts'' towards the direction indicated by the new stream item. C-style pseudo code of this algorithm is described in Algorithm~\ref{alg:1umedian}, \oneUmedian.

\begin{example}
\textit{To illustrate how \oneUmedian works, let us consider the example in Figure~\ref{fig:stream_example}.
For the first 2 stream items \{$s_1$ = 4, $s_2$ = 2\} the stream median $m_S$ is 4, when the third item $s_3=1$ comes, the stream median $m_S$ becomes 2. The estimated median from \oneUmedian algorithm starts from $\tilde{m}_0=0$, and gets updated on each arriving stream item. For example, when $s_4 = 5$ comes, it is larger than $\tilde{m}_3$ whose value is 1, therefore $\tilde{m}_4 = \tilde{m}_3+1 = 2$. In this example, $\tilde{m}$ starts from 0, and after reading 5 items from the stream it reaches the stream median for the first time.}
\label{ex:continuous_example}
\end{example}

\begin{figure}[tb!]
\centering
\epsfig{file=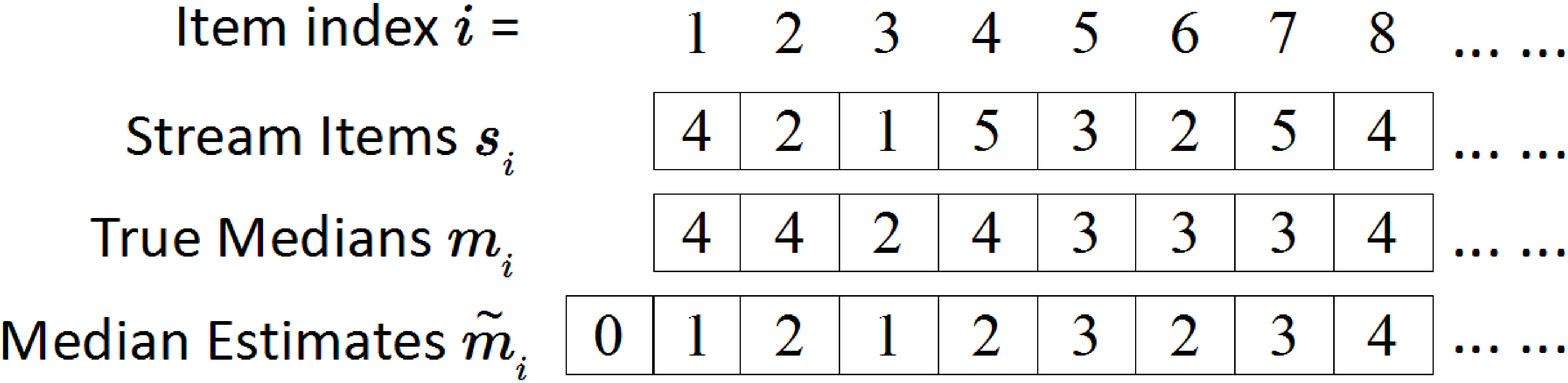, height=0.8in, width=80.7mm}
\vspace{-0.1in}
\caption{Estimate stream median}
\label{fig:stream_example}
\end{figure}

In Example ~\ref{ex:continuous_example}, values in the stream are contiguous without gaps. So the approximations from \oneUmedian algorithm can give accurate value $\frac{1}{2}$-quantiles, and $\tilde{m}_5$, $\tilde{m}_7$ and $\tilde{m}_8$ are correct approximations for stream medians.
Let us look at another example below where \oneUmedian algorithm gives accurate estimates in terms of rank $\frac{1}{2}$-quantile approximation.

\begin{example}
\textit{In Figure~\ref{fig:discrete_example}, the stream median is $10$ after seeing the first $2$ items. \oneUmedian gives median approximations $2$ or $3$ after updating on those two items. Although $2$ or $3$ are not in the value domain of this stream, it satisfies the rank $p$-quantile definition.}
\label{ex:discrete_example}
\end{example}

\begin{figure} [tb!]
\centering
\epsfig{file=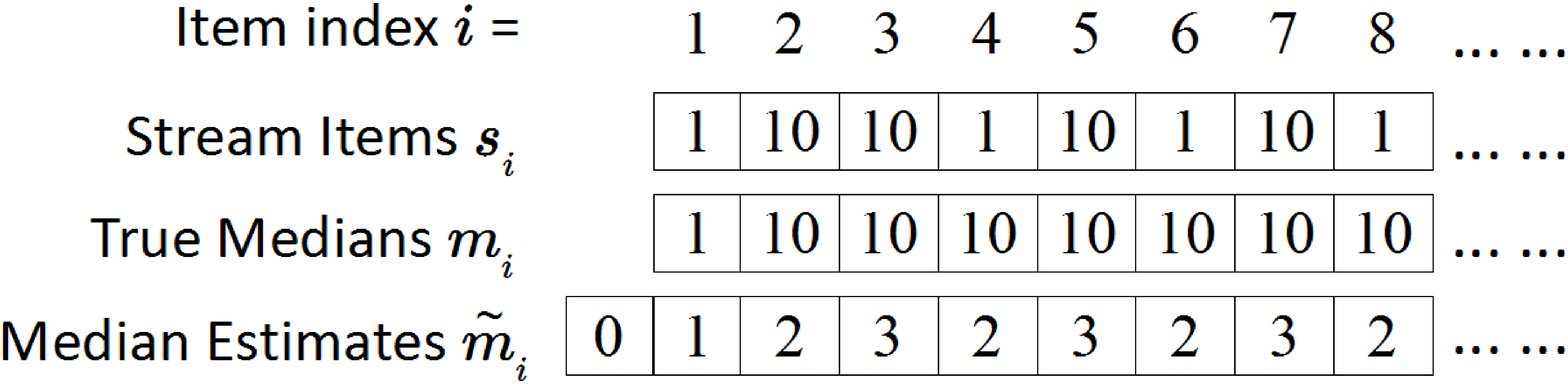, height=0.8in, width=80.7mm}
\vspace{-0.1in}
\caption{Stream from a gapped  domain}
\label{fig:discrete_example}
\vspace{-0.15in}
\end{figure}

\subsection{1 Unit Memory to Estimate Any Quantile}
\label{subsec:1memquantile}

Following the same intuition as above, we can use $1$ unit memory to estimate any $\frac{h}{k}$-quantile, where $\mathit{1 \leq h \leq k-1}$. If current stream item is larger than estimation, we need to increase estimation by $1$; otherwise, we need to decrease estimation by $1$. Up to this point it is the same as \oneUmedian  algorithm. The trick to generalize to $\frac{h}{k}$-quantile is that not every stream item seen will cause an update. If current stream item is larger than estimation, an increment update will be triggered only with probability $\frac{h}{k}$. The rationale behind it is that if we are estimating $\frac{h}{k}$-quantile, and if current estimation is at stream's true $\frac{h}{k}$-quantile, we will expect to see stream items larger than current estimation with probability $1-\frac{h}{k}$. If the probability of seeing larger stream items is greater than $1-\frac{h}{k}$, it is caused by the fact that current estimation is smaller than stream's true $\frac{h}{k}$-quantile. Similarly, a smaller stream item will cause a decrement update only with probability $1-\frac{h}{k}$. Our general $1$ unit memory quantile estimation algorithm is described in Algorithm~\ref{alg:1uquantile}, \oneUany.

We need to make a few observations from this algorithm. Besides $\tilde{m}$, this algorithm uses $rand$ and $\frac{h}{k}$. Notice that we can implement the algorithm without explicitly storing $rand$ value,
$\frac{h}{k}$ is a constant across all the groups, no matter how many, and can be kept in registers.

\begin{algorithm}[tb!]
\caption{\oneUany}
\label{alg:1uquantile}
\begin{algorithmic}[1]
\REQUIRE Data stream $S$, $h$, $k$, 1 unit of memory $\tilde{m}$
\ENSURE $\tilde{m}$
\STATE Initialization $\tilde{m} = 0$ 
\FOR {\textbf{each} $s_i$ in $S$}
	\STATE $rand$ = random(0,1); \COMMENT get a random value in [0,1]
  \IF{$s_i > \tilde{m}$ \AND $rand > 1 - \frac{h}{k}$}
   		\STATE $\tilde{m} = \tilde{m} + 1$;
  \ELSIF{$s_i < \tilde{m}$ \AND $rand > \frac{h}{k}$}
   		\STATE $\tilde{m} = \tilde{m} - 1$;
  \ENDIF
\ENDFOR	
\end{algorithmic}
\end{algorithm}

Update taken by $\tilde{m}$ in Algorithm~\ref{alg:1uquantile} is 1, it is small change at each step when the stream quantile to estimate is large. When it is allowed one extra unit of memory, we can use it to store the size of update to take, denoted as $step$. Extension to two unit memory algorithm is to be presented in Section \ref{sec:extension}. 

\section{Analysis}
\label{sec:analysis}

Our frugal algorithm for estimating a quantile can be arbitrarily bad on worst case streams. This is expected because our algorithm has no memory of the past. One type of such worst case streams is that the true stream quantile value to be estimated has high probably in its underlying distribution. Therefore even if current estimation is at true stream quantile, a minimum update of 1 to quantile estimation will cause large change in rank quantile error. Also any adversary can remember the entire past and constantly mislead our algorithm.
For example, the order of stream items can affect the estimation.

\begin{example}
\textit{In this example, Figure \ref{fig:bad_example}, stream items are in ascending order. Median estimation of \oneUmedian, $\tilde{m}$, starts from value 0. Every $s_i$ is larger than $\tilde{m}_{i-1}$, so that $\tilde{m}$ gets increased on very item. These median approximations are incorrect since they do not give correct value or rank quantile estimations.}
\end{example}

Indeed, any frugal streaming algorithm for any problem is likely to face such lower bounds. The real intuition and strength of our algorithm comes from elsewhere. We say a stream is {\em Stochastic} if each stream item is drawn from some
distribution $D$, independently and randomly from other stream items. We will analyze and show that our algorithm quickly converges to an estimate of the target quantile, and further, stably remains in the neighbourhood of the quantile as stream progresses.

\begin{figure}[tb!]
\centering
\epsfig{file=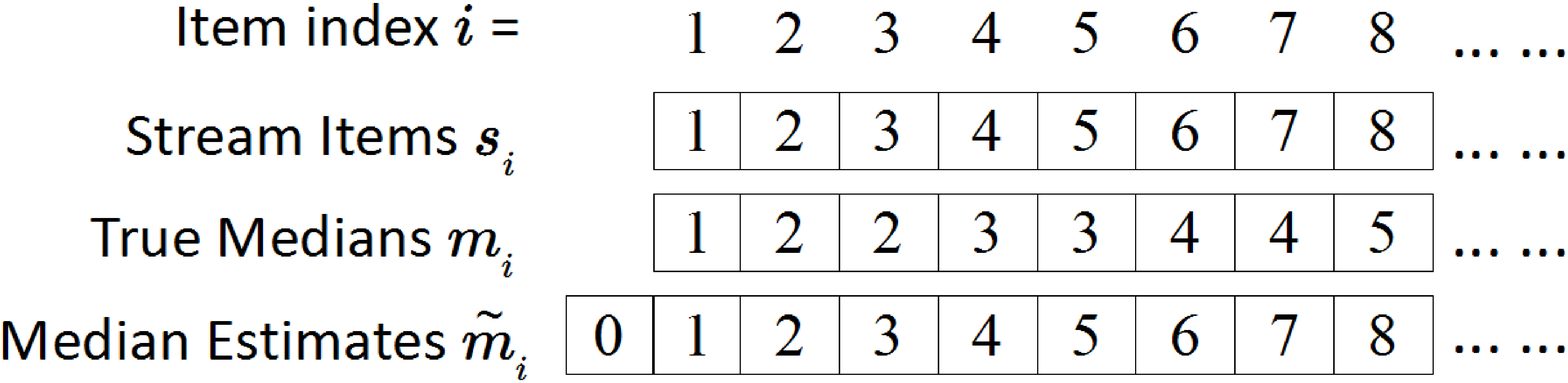, height=0.8in, width=80.7mm}
\vspace{-0.1in}
\caption{Stream items in ascending order}
\label{fig:bad_example}
\vspace{-0.15in}
\end{figure}

\subsection{\textbf{Approaching Speed}}
\label{subsec:approach}

For our 1 memory algorithm, each update size is 1. At any time $t_i$, our algorithm estimation has non-zero probabilities to move towards or away from true quantile. Therefore for sufficiently large $t$, the probability that algorithm estimation moves continuously in one direction has very low probability. When current algorithm estimation is far away from true quantile, the speed of approaching true quantile is high, since every update is highly biased towards true quantile. But as the estimation gets closer to true quantile, the bias to move towards true quantile gets weaker so the speed of approaching true quantile is low. In other words, we are likely to see algorithm estimation showing an oscillating trajectory towards true quantile. The analysis of our algorithm is non-trivial and challenging because the rate of the convergence to an estimate is not constant and depends on number of varying factors. We rely on the concept of stochastic dominance and we show that in fact the algorithm will approach the true quantile with linear speed.

\newcommand{\vareps}{\varepsilon}
\newcommand{\cdf}{F}
\newcommand{\qf}[1]{Q(#1)}
\renewcommand{\Pr}[1]{{\bf Pr}\left[#1\right]}
\newcommand{\Ex}[1]{{\bf E}\left[#1\right]}
\newcommand{\ploc}[1]{y(#1)}
\newcommand{\unknown}{?}
\newcommand{\floor}[1]{\left\lfloor#1\right\rfloor}
\newcommand{\xx}{x}
\newcommand{\XX}{\tilde{m}}
\newcommand{\half}{\frac{1}{2}}
\newcommand{\PROOF}{\noindent \emph{Proof:} }

Recall our notations from Section~\ref{sec:pre}, $\cdf(t)$ is the $CDF$ of distribution, $\qf{x}$ is quantile. Let $\xx_i$ be an indicator variable for the direction of $i$-th step of the algorithm, where $\xx_i=1$ for increment and $\xx_i=-1$ for decrement. Let $\XX_t = \sum_{i=1}^t x_i$. In other words $\XX_t$ is the estimation of the quantile at time $t$. Assume $|\cdf(i) - \cdf(i+1) |\le \delta$, so $\delta$ is the maximum single location probability in distribution and $0\leq \delta < 1$.
Let algorithm estimate $\frac{h}{k}$ quantile,
whose value is denoted as $M$. Assume algorithm estimation starts from position $\XX_0$, where $\XX_0 < M$. The distance from starting position to true quantile is $M - \XX_0$, but the analyses trivially generalize to the case where the distance to the true quantile is $M$.

\begin{lemma}
For median estimation, assume algorithm estimation starts from position $\XX_0$, where $F(\XX_0) < \frac{1}{2} - \delta$.
After $T = \frac{M|\log {\vareps}|}{\delta}$ steps of algorithm, the probability that $\cdf(\XX_t) < \frac{1}{2} - \delta$ for all $t < T$  is at most $\vareps$. In other words, after $O(M)$ steps, with high probability
the algorithm has crossed vicinity of the true quantile, $\frac{1}{2}-\delta$, at least once.
\label{lemma:left}
\end{lemma}

\begin{proof}
Let $M' = \qf{\frac{1}{2} - \delta}$. Let us compute the expectation of a move whenever the algorithm is below $M'$.
$$\Pr{x_i=1}  \ge \frac{1}{2} ( 1 - (\frac{1}{2} - \delta) ) = \frac{1}{2} - \frac{1}{2^2} +  \delta* \frac{1}{2} $$ 
we denote it by $\theta$, 
then
$$\Pr{x_i = -1 } \le (1 - \frac{1}{2}) (\frac{1}{2} - \delta) = \theta - \delta$$
Therefore we have
\begin{equation}
\label{eq:Exxi}
\Ex{x_i} \ge \delta 
\end{equation}
  In other words the expected shift of each $x_i$ before it hits $M'$ is then at least $\delta$.
To prove our lemma, we therefore can use tail inequalities to bound the deviation of $\XX_t = \sum x_i$ from the expectation. The main difficulty, however arises from the
fact $x_i$ are not independent from each other and the constraint \eqref{eq:Exxi} holds only when $\XX_t \le M'$.
Consider an arbitrary sequence of moves $x_i$. Define $y_i = x_i$ for all $i < i_0$, where $i_0$ is the time where $\XX_{i_0}$ crossed $M'$ for the first
time, and $y_i = 1$ with probability $\theta$, $y_i = -1$ with probability $\theta - \delta$, and 0  otherwise.
Similarly we define $Y_t=\sum y_i$ for all $i < i_0$, where $i_0$ is the time where $\XX_{i_0}$ crossed $M'$ for the first
time. Then we have $\Pr{\XX_i < M' ,~\forall i \in [T]} = \Pr{Y_i < M' ,~\forall i \in [T]}$. Therefore it is enough for us to prove our statement
for $Y_i$. However, $Y_i$ are still not necessarily independent from each other, before they cross $M'$, however all of them satisfy $\Ex{y_i} \ge \vareps$ and
$\Pr{y_i = 1} \ge \theta$, and $\Pr{y_i = -1} \le \theta - \delta$.
 Define $z_i$ (and $Z_i$ respectively), such that $z_i$ is stochastically dominated by $y_i$ and each $z_i$ is 1 with probability $\theta$ and $-1$ with probability
$\theta - \delta$. Using Hoeffding inequality we have:
$$
\Pr{|Z_t - \Ex{Z_t}| > C}  \le \exp ({-\frac{tC}{2}})
$$
using the fact $$\Ex{Z_t} \ge \delta  t \ge M |\log \vareps| = M  - (M \log \vareps + M)$$ and using $C = (M + M\log \vareps)$ and using union bound over all $t$ we have desired result immediately
for $Z_t$, using the fact that $Y_t \ge Z_t$ we have that probability $Y_t$ never crosses the bound is less than $\vareps$ and hence lemma holds.
\end{proof}

Note, that our constraints are spelled in terms of probability mass inequality rather than absolute error. This is required, since  for any function $f(M)$, it is possible to devise a distribution, such that
the algorithm will be $f(M)^2$ far away from true quantile in absolute steps, and yet it will be very close to it in terms of probability mass.

\begin{lemma}
\label{lemma:right}
For median estimation, algorithm estimation starts from a position $\XX_0$, where $F(\XX_0) > \frac{1}{2} + \delta$.
After $T=\frac{M|\log {\vareps}|}{\delta}$ steps of algorithm, the probability that $\cdf(\XX_t) > \frac{1}{2} + \delta$ for all $t < T$  is at most $\vareps$. \omt{In other words, after $O(M)$ steps, with high probability
the algorithm has crossed vicinity of the true quantile at least once.}
\end{lemma}
\begin{proof}
Proof is similar to Lemma \ref{lemma:left}.
\end{proof}

\begin{theorem}
\label{thm:approach}
For median estimation, algorithm estimation starts from a position $\XX_0$, where $F(\XX_0)$ is outside of region $[\frac{1}{2} - \delta, \frac{1}{2}+ \delta]$.
After $T=\frac{M|\log {\vareps}|}{\delta}$ steps the algorithm, the probability that $\cdf(\XX_t)$ is outside of this close region $[\frac{1}{2} - \delta, \frac{1}{2}+ \delta]$ for all $t < T$  is at most $\vareps$. \omt{In other words, after $O(M)$ steps, with high probability
the algorithm has crossed vicinity of the true quantile at least once.}
\end{theorem}
\begin{proof}
Proof is directly obtained from Lemma~\ref{lemma:left} and Lemma~\ref{lemma:right}.
\end{proof}

\omt{
\begin{lemma}
\label{lemma:right}
Algorithm estimation starts from a position larger than $\alpha + \delta$.
After $\frac{M|\log {\vareps}|}{\delta}$ steps of algorithm, the probability that $\cdf(X_t) > \alpha + \delta$ for all $t < T$  is at most $\vareps$. \omt{In other words, after $O(M)$ steps, with high probability
the algorithm has crossed vicinity of the true quantile at least once.}
\end{lemma}

\begin{theorem}
\label{thm:approach}
Algorithm estimation starts from a position outside of a close region $[\alpha - \delta, \alpha + \delta]$.
After $\frac{M|\log {\vareps}|}{\delta}$ steps the algorithm, the probability that $\cdf(X_t)$ is outside of the close region for all $t < T$  is at most $\vareps$. \omt{In other words, after $O(M)$ steps, with high probability
the algorithm has crossed vicinity of the true quantile at least once.}
\end{theorem}
\begin{proof}
Proof is directly obtained from Lemma~\ref{lemma:left} and Lemma~\ref{lemma:right}.
\end{proof}
}

In approaching speed analysis, we do not need assumptions on algorithm's starting estimation. Therefore this actually implies for \oneUany algorithm, quantile estimations adjust to new distribution quantile when underlying distribution changes, regardless of current estimation position. The speed of approaching new distribution quantile \omt{and the expected number of steps needed} can be determined by Theorem~\ref{thm:approach}.

\subsection{\textbf{Stability}}

Next we show that after algorithm estimation once reaches true median, the probability of estimation drifting far away from true median is low. Note that THEOREM \ref{thm:approach} is affecting this estimation drifting process the whole time.

\begin{lemma}
For median estimation, assume current estimation is at true median. After $t$ steps, the probability of the algorithm current position
$$\Pr{\cdf(\XX_t) > \frac{1} 2 +  2\sqrt{\delta \ln \frac{t}{\vareps}}} \le \vareps.$$
\label{lemma:stability-right}
\end{lemma}

\begin{proof}	
Define $\omega = 2\sqrt{\delta \ln \frac{t}{\vareps}}$. Let us split the interval $[\frac{1}{2}, \frac{1}{2} + \omega ]$ into
two $[\frac{1}{2}, \frac{1}{2} + \omega/2]$ and $[\frac{1}{2} + \omega/2, \frac{1}{2} + \omega]$. Our approach is to show
that once the algorithm reaches the boundary of the first interval, it is very unlikely to continue through the second interval,
without ever dipping back into the first. First of all we note that we need at least $T = \frac{\omega}{\delta}$ more steps of increment than decrement to reach outside of the second interval, and by the way we select the probabilistic weight of the interval,
we will need at least $T/2$ to pass through each.

Consider arbitrary outcome of the algorithm where $\XX_t > T$. Since $x$ changes by at most 1 at every
step, there exists $j$, such that $\XX_j = \frac{T}{2}$. Therefore the entire space of events
can be decomposed based on the value of $j$ where $\XX_j = \floor{T/2}$ and for all $i>j$,
$\XX_i > \XX_j$. Thus:
$$
\begin{array}{rl}
\Pr{\XX_t > T} &=\sum\limits_{j=0}^t \Pr{\XX_t > T,\XX_i> \XX_j,\forall i>j} \\
& \times \Pr{\XX_j = \floor{\frac{T}{2}}} \\
	       &\le\sum\limits_{j=0}^t \Pr{\XX_t > T,\XX_i> \XX_j,\forall i>j}
\end{array}\
$$
let us consider individual term for a fixed $j$ in the sum above. We want to show that each term
is at most $\vareps/t$.
Define $Y^{(j)}_i$ for $i\ge j$, where $Y^{(j)}_i = \XX_j + \sum_{k=j+1}^i y_j$, and
 $y^{(j)}_i = x_i$ if $X_i' > \XX_j$, for all $i' < i$,
and for the remainder of the segment $y^{(j)}_i$ is random variable that is -1
with probability $p = \half + \frac{\omega}2$ and 1 otherwise.
In other words $Y_i$ agrees with $\XX_i$ until $\XX_i = \XX_j$ for the
first time after $j$, after that $Y^{(j)}_i$ becomes independent of $\XX_i$. We have:
$$
\begin{array}{rl}
& \Pr{\XX_t > T, \XX_i > \XX_j, \forall i>j}  \\
&= \Pr{Y^{(j)}_t > T, Y^{(j)}_i > Y^{(j)}_j, \forall i>j}  \\
& \le \Pr{Y^{(j)}_t > T}
\end{array}
$$
therefore it is sufficient to compute an upper bound for  $\Pr{Y^{(j)}_t > T}$ for all $j$. Let $Z^{j}_i$  be
a variable which both stochastically dominates  $Y^{(j)}_i$, and is -1 with probability $p$
and $1$ otherwise. Since $Y^{(j)}_i$ is $-1$ with probability of \emph{ at least} $p$,
so such variable always exists.  Note that $Z^{j}_i$ are independent from each other for all $i$, thus we can use standard \omt{issue} tail inequality to upper bound $Z^{(j)}_t$, and because of the dominance the result will immediately apply to $Y^{(j)}_i$.
Since $Z^{(j)}_i$ only depends on $j$ at the starting point, we can shift it to zero and rewrite out constraint as:
$$
  \sum_{j=0}^{t} \Pr{Z_j > T/2} \le \vareps
$$
where $Z_j$ is defined as sum $\sum_{i=0}^j z_i$, and $z_i$ is -1 with probability $p$ and $1$ otherwise. The expected value of
$Z_j$ is $(1-p) j - p j = (1-2p)j = -\omega j$. Furthermore by our assumption, $\omega \ge \frac{\delta T}{2}$.
 Therefore using Hoeffding inequality we have $\Pr{Z_j > T/2} \le \exp{-\frac{(\omega j + T)^2}{4j }}$.
Thus it is sufficient for us to show that
$$
\exp{-\frac{(\omega j + T)^2}{4j}} \le \frac{\vareps}{t}, \mbox{ for all } j < t
$$
This constraint is automatically satisfied for all $j$ such that
$$j \ge \frac{4}{\omega^2}\ln \frac{t}{\vareps}=j_0.$$ Indeed, if $j > j_0$  we have
$(\omega j+T) / 4j \ge \frac{\omega^2}{4j} \ge \ln t/\vareps$.

On the other hand if $j \le j_0$, then we have
$$
\frac{(\omega j + T)^2} {4j} \ge \frac{T^2 \omega^2}{16\ln t/\vareps}
$$
but $T \ge \omega / \delta$ and substituting the expression for $\omega$ we have:
$$
\frac{T^2 \omega^2}{4\ln t/\vareps} \ge \frac{\omega^4}{16 \delta^2 \ln t/\vareps}  = \ln t/\vareps
$$

Thus $\Pr{Z_j > T/2} \le \vareps/t$, for $j < j_0$, completing the proof.
\end{proof}

\begin{lemma}
\label{lemma:stability-left}
To estimate median, assume current estimation is at true median. After $t$ steps, the probability of the algorithm current position
$$\Pr{\cdf(\XX_t) < \frac{1}{2} -  2\sqrt{\delta \ln \frac{t}{\vareps}}} \le \vareps.$$
\end{lemma}

\begin{proof}
Following the same reasoning in the proof of LEMMA~\ref{lemma:stability-right}, we can prove that the probability of estimation moving far to the left is small. Where we can split the interval $[\frac{1}{2}-\omega, \frac{1}{2}]$ into two $[\frac{1}{2} - \omega, \frac{1}{2} - \omega/2]$ and $[\frac{1}{2} - \omega/2, \frac{1}{2}]$. We can show that once the algorithm reaches the boundary of the first interval, it is very unlikely to continue through the second interval without ever dipping back into the first.
\end{proof}

\begin{theorem}
To estimate median, assume current estimation is at true median. After $t$ steps, the probability of the algorithm current position
$$\Pr{ \left|\cdf(\XX_t) - \frac{1}{2}\right| >  2\sqrt{\delta \ln \frac{t}{\vareps}}} \le \vareps.$$
\end{theorem}
\begin{proof}
This theorem is directly obtained from Lemma~\ref{lemma:stability-right} and ~\ref{lemma:stability-left}.
\end{proof}

These properties of median estimation can be generalized to any quantile $\frac{h}{k}$.
\section{Algorithm Extensions}
\label{sec:extension}

\begin{algorithm}[tb!]
\caption{\twoUany}
\label{alg:2uquantile}
\begin{algorithmic}[1]
\REQUIRE Data stream $S$, $h$, $k$, $\tilde{m}$, $\tt{step}$, $sign$
\ENSURE $\tilde{m}$
\STATE Initialization $\tilde{m} = 0$, $\tt{step} = 1$, $sign=1$
\FOR {\textbf{each} $s_i$ in $S$}
  \STATE $rand$ = random(0,1);
  \IF{$s_i > \tilde{m}$ \AND $rand > 1 - h/k$}		
  		\STATE $\tt{step}$ $\mbox{ += } (sign>0)~?~f($$\tt{step}$$) ~: ~-f($$\tt{step})$;
   		\STATE $\tilde{m} \mbox{ += } ($$\tt{step}$$ >0) ~?~\lceil $$\tt{step}$$ \rceil ~: ~ 1$;
	     \IF {$\tilde{m} > s_i$}
	     	\STATE $\tt{step}$$ \mbox{ += } s_i - \tilde{m}$;
	      	\STATE $\tilde{m} = s_i$;
	     \ENDIF
	    \IF{$sign < 0$ \AND $\tt{step}$$>1$}
	     		\STATE $\tt{step}$$ = 1$;
	    \ENDIF
	    \STATE $sign = 1$;
  \ELSIF{$s_i < \tilde{m}$ \AND $rand > h/k$}
  	   	\STATE $\tt{step}$$\mbox{ += } (sign < 0) ~? ~ f($$\tt{step}$$) ~:~ -f($$\tt{step}$$)$;
    	\STATE $\tilde{m} \mbox{ - = } ($$\tt{step}$$ > 0) ~?~ \lceil $$\tt{step}$$ \rceil ~:~ 1$;
        \IF {$\tilde{m} < s_i$}
        	\STATE $\tt{step}$$ \mbox{ += } \tilde{m} - s_i$;
        	\STATE $\tilde{m} = s_i$;
        \ENDIF
        \IF{$sign > 0$ \AND $\tt{step}$$>1$}
        	\STATE $\tt{step}$$ = 1$;
        \ENDIF
        \STATE $sign = -1$;
  \ENDIF
  
\ENDFOR	
\end{algorithmic}
\end{algorithm}

The \oneUany algorithm described in Section \ref{sec:algs} uses 1 unit of memory and is intuitive, and we managed to analyze it; however  it has linear convergence to the true quantile. This is effectively by design, because the algorithm does not have the capability to remember anything except the current location. A simple extension to our algorithm is to keep a current step size in memory, and modify it if the new samples are consistently on one side of the current estimate.\footnote{Another
approach that we do not explore here, is to use multiplicative update on step size instead of additive.}
In this section we describe a $2$ units of memory algorithm that we use in experiments for comparison. 

Generally the algorithm uses two variables to keep quantile estimate and update size, and one extra bit to keep sign, which indicates the increment or decrement direction of estimate. Empirically this  algorithm has much better convergence and stability 
property than 1 unit of memory algorithm, however the precise convergence/stability analysis of it is one of our future work. On the intuitive level the algorithm for finding the median works as follows. As before it maintains the current estimate of median but in addition it also maintains an update $\tt{step}$ that increases or decreases based on the observed values, determined by a function $f$. More precisely, the $\tt{step}$ increases if the next element from the stream is on the same side of the current estimate, and decreases otherwise. When estimation is close to true quantiles, $\tt{step}$ can be decreased to extremely small value. 

The increment and decrement factors to be applied to $\tt{step}$ remains an open problem. $\tt{step}$ can potentially grow to very large values, so the randomness of the order which stream items appear affects estimation accuracy. For example, if let $\tt{step}_i$ be the step value at $i$th update, a multiplicative update of $\tt{step}_{i+1}=2\times \tt{step}_{i}$ might be a good choice for a random order stream, which intuitively needs $O(\log M)$ updates to reach true quantile at distance $M$ from current estimate. However in empirical data periodic pattern might be apparent in the stream, for example social network users might have shorter activity intervals at evening, but longer intervals at early morning. Then $\tt{step}$ can easily get increased to a huge value. It will make the algorithm estimate drift far away from true quantile, hence estimates will have large oscillations. 

Therefore to trade off convergence speed for estimation stability we present a version of 2 units of memory algorithm that applies constant factor additive update to step size, where $f(\tt{step})=1$. Full details of the algorithm are described in Algorithm~\ref{alg:2uquantile}. Lines 4-14 handle stream items larger than algorithm estimation, and lines 15-26 handle smaller stream items. For brevity we only look at lines 4-14 in detail.
Similar to Algorithm \oneUany, the key to make \twoUany able to estimate any quantile is that not every stream item will cause an estimation update, so line 4 enables updates only on ``un-expected'' larger stream items. $\tt{step}$ is cumulatively updated in line 5. Line 6 ensures minimum update to estimation is 1, and $\tt{step}$ size is only applied in update when it is positive. The reason is that when algorithm estimation is close to true quantile, \twoUany updates are likely to be triggered by larger and smaller (than estimation) stream items with largely equal chances. Therefore $\tt{step}$ is decreased to a small negative value and it serves as a buffer for value bursts (\eg~a short series of very large values) to stabilize estimations. Lines 7-10 are to ensure estimation do not go beyond empirical value domain when $\tt{step}$ gets increased to very large value. At the end of the algorithm, we reset $\tt{step}$ if its value is larger than 1 and two consecutive updates are not in the same direction. This is to prevent large estimate oscillations if $\tt{step}$ gets accumulated to a large value. This checking is implemented by lines 11-13.

Note that \oneUany and \twoUany algorithms are initialized by 0, but in practice they can be initialized by the first stream item to reduce the time needed to converge to true quantiles.

\section{Related Work and Algorithms to Compare}
\label{sec:algstocompare}

There has been extensive work in the database community on theory and practice of approximately estimating quantiles of streams with limited memory (e.g..,~\cite{Greenwald01onlinequantile, Babcock03maintainvarandmedian, Arasu04approxcounts, Lin04quantilesummaries, Cormodbiasedquan, Agrawal95aone-pass,AlsabtiRanka-1997, Manku98approxs, Guha_streamorder, Gilbert20021287369, Shrivastava04mediansandbeyong, Cormode200558}). This body of research has generated methods for
approximating quantiles to $1+\epsilon$ approximation with space roughly $O(1/\epsilon)$ in various models of data streams.

We compare our algorithms with existing algorithms that use constant memory for stochastic streams \cite{Guha_streamorder}, and also non-constant memory algorithms described in \cite{Greenwald01onlinequantile, Shrivastava04mediansandbeyong}.
However all the non-constant memory algorithms above use considerably more than 2 persistent variables. While some of the algorithms such as the one described in \cite{Agrawal95aone-pass} have a tuning parameter allowing to decrease memory utilization, the algorithm then performs poorly when used with less than 20 variables. Here we briefly overview the algorithms we compare with.

\subsection{\gk Algorithm}
Greenwald and Khanna \cite{Greenwald01onlinequantile} proposed an online algorithm to compute $\epsilon$-approximate quantile summaries with worst-case space requirement of $O(\frac{1}{\epsilon}log(\epsilon N))$. Greenwald-Khanna algorithm (\gk) maintains a list of tuples ($v_i$, $g_i$, $\vartriangle_i$), where $v_i$ is a value seen from the stream and tuples are order by $v$ in ascending order. $\sum_{j=1}^{i}g_j$ gives the minimum rank of $v_i$, and its maximum rank is $\sum_{j=1}^{i}g_j + \vartriangle_i$. \gk is composed of two main operations which are to insert a new tuple in to tuple list when sees a new value, and do compression on the tuple list to achieve the minimum space as possible. Throughout the updates it is kept invariant that for any tuple we have $\sum_{j=1}^{i}g_j + \vartriangle_i \leq 2\epsilon N$ to ensure the $\epsilon$-approximate query answers. The main difference of our algorithms is that our scenarios do not require the ability to answer any quantile queries, but only a few quantiles are of interest. Hence our advantage is saving space usage by not tracking non-necessary quantiles. In the original \gk algorithm desired $\epsilon$ is accepted as input, and it will use as less space as possible to achieve $\epsilon$-approximate. To make it comparable with our \oneUany and \twoUany, we limit the number of tuples maintained by \gk. When this memory budget is exceeded we gradually increase $\epsilon$ (increment by 0.001) to force compression operation get conducted repeatedly until number of tuples used is within specified budget. In our comparison, we limit the number of tuples to be $t=20$.

\begin{figure*}[t]
\centering
\begin{tabular}{cc}
\epsfig{file=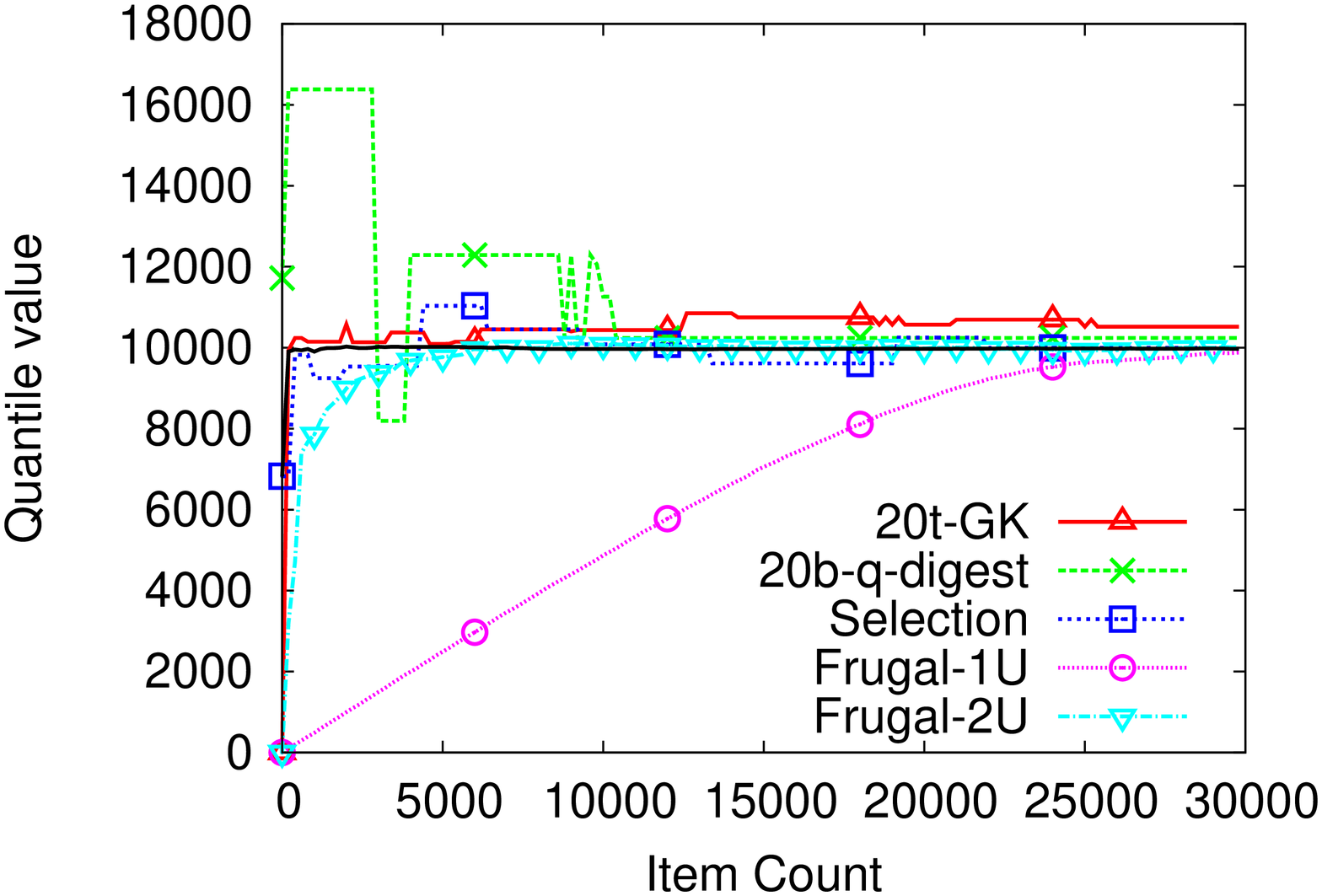, width=\figurewidthJ, angle=-90} &
\epsfig{file=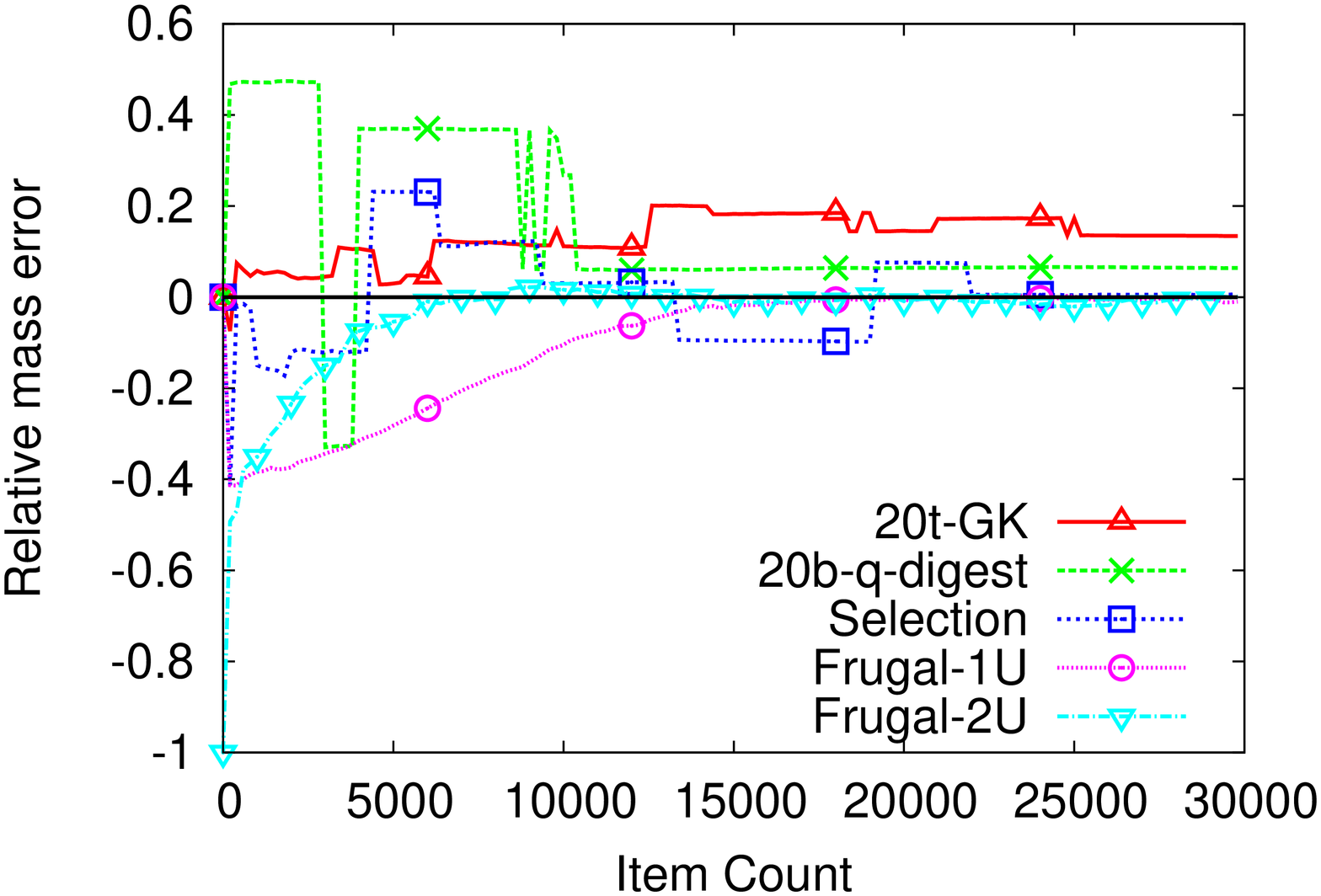, width=\figurewidthJ, angle=-90} \\
\vspace{-0.1in}
\mbox{(a)} & \mbox{ (b)} \\
\epsfig{file=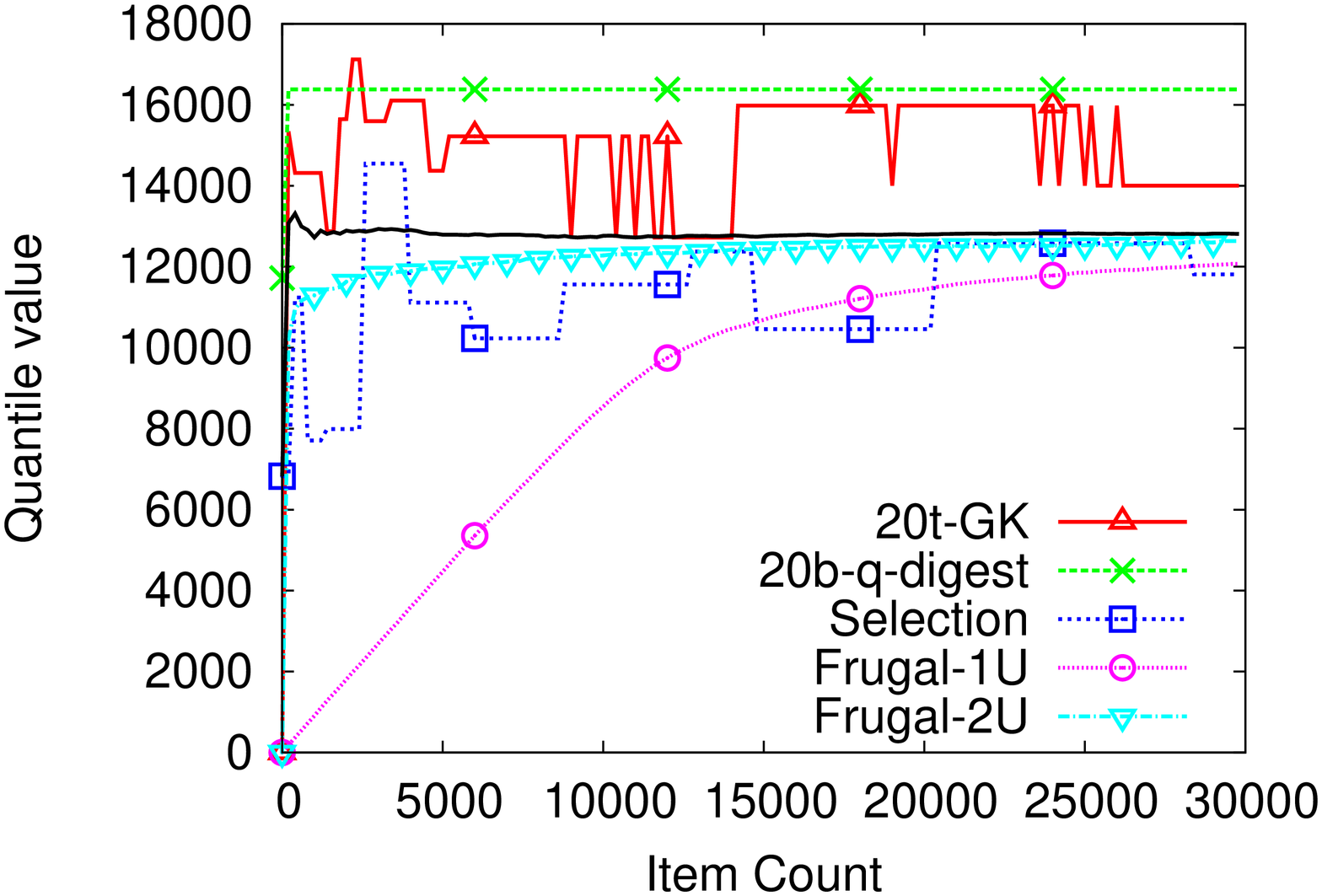, width=\figurewidthJ, angle=-90} &
\epsfig{file=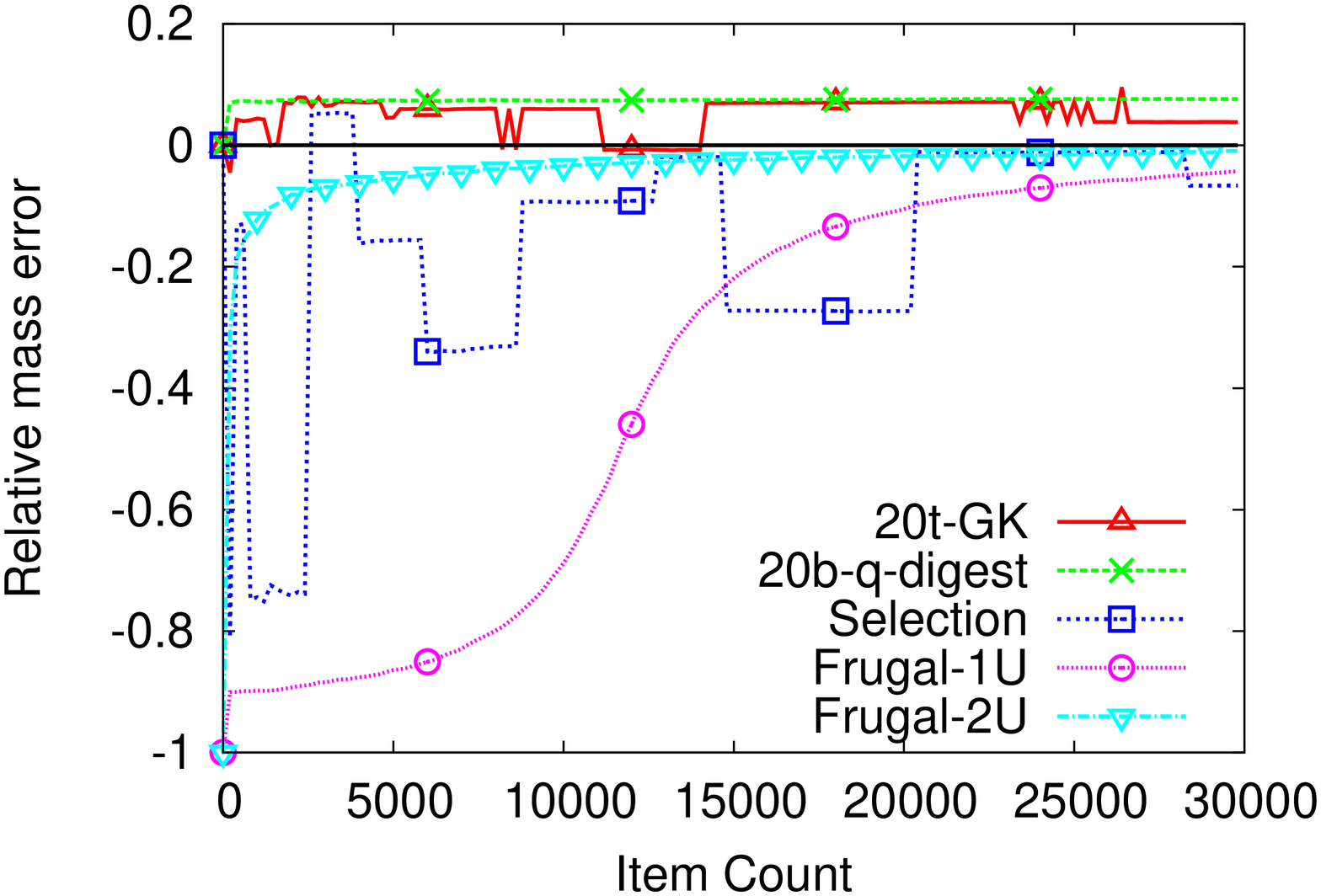, width=\figurewidthJ, angle=-90} \\
\mbox{ (c)} & \mbox{ (d)} \\
\end{tabular}
\vspace{-0.05in}
\caption{Evaluation on stream from one Static Cauchy Distribution. (a) median estimation. (b) relative mass error for (a). (c) 90-\% quantile estimation. (d)  relative mass error for (c).}
\label{plot:static_median}
\vspace{-0.15in}
\end{figure*}

\subsection{\qdigest Algorithm}
Tree based stream summary algorithms were studied by Manku et al. \cite{Manku98approxs}, Munro and Paterson \cite{MunroPaterson-1980}, Alsabti et at. \cite{AlsabtiRanka-1997}, Shrivastava et al. \cite{Shrivastava04mediansandbeyong} and Huang et al. \cite{Huang11qdigestlesscommunication}. In this paper we compare with \qdigest algorithm proposed in \cite{Shrivastava04mediansandbeyong}, which is up to date and most relevant to our comparison aspects. Their proposed algorithm builds a binary tree on a value domain $\sigma$, with depth $log\sigma$. Each node $v$ in this tree is considered as a bucket representing a value range in the domain, associated with a counter indicating the number of items falling in this bucket. A leaf node represents a single value in domain, and associated with the number of items having this value. Each parent node represents the union of the ranges of children nodes, root node represents the full domain range. This algorithm then keeps merging and removing nodes in the tree based on the following two conditions:
\begin{align}
count(v) &\leq \left\lfloor \alpha \right\rfloor \\
count(v) + count(v_p) + count(v_s) &> \left\lfloor \alpha \right\rfloor
\end{align}
Where $v_p$ is the parent and $v_s$ is the sibling of $v$,
and $\alpha$ is chosen based on memory constraints. If a non-leaf node violates the second constraint, its children are merged into $v_p$, and $v$ and $v_s$ are deleted. The original application of this algorithm was to sensor network, however authors also proposed an adaptation to streaming which is the variant we consider here. For every new stream sample we make a trivial \qdigest and merge it with \qdigest built so far. Therefore, at any time we can query for a quantile based on the most recently updated \qdigest. For our evaluation we used number of buckets of $b=20$ 
to build tree digests, presenting the case where insufficient memory are used. Note that empirically the used memory is usually larger than the budget specified, because conditions (1) and (2) do not always guarantee a bucket will be freed up when insertion of a new item is needed. As pointed out in the paper, the actually used memory might be more than specified $b$ while no more than 3$b$.
We refer to this algorithm as \qdigest in our comparisons, and stream domain maximum value is given as required input at the beginning in order to build a binary tree for digest generation.

\subsection{\selection Algorithm}

Guha and McGregor \cite{Guha_streamorder} proposed an algorithm that uses constant memory and operates on random order streams, where the order of elements of the stream have not been chosen by adversary. Their approach is a single pass algorithm that uses constant space and their guarantee is that for a given $r$ (the rank of element of interest) their algorithm returns an element that is within  $O(n^{1/2})$ rank of $r$ with probability at least $1-\delta$ if the stream is randomly ordered. The algorithm does not require prior knowledge of the length of the stream, nor the distribution, which are also not required by our \oneUany and \twoUany.

This single-pass algorithm ($Selection$) processes the stream in phases, and each phase is composed of three sub-phases namely, \textit{sample}, \textit{estimate} and \textit{update}. Throughout the process, algorithm maintains an interval $(a, b)$ which encloses the true quantile. Each phase is trying to narrow this interval. In \textit{sample} phase, a $u$ in $(a, b)$ is selected and get its rank estimated in \textit{estimate} phase, lastly $a$ or $b$ might be replaced by $u$ in \textit{update} phase based on the estimated rank of $u$ above or below true quantile. To work with these three sub-phases, stream is divided into pieces and each piece is used for one phase. Then each piece of the stream is divided into two parts, first part is used for \textit{sample} sub-phase, and the second piece is used for \textit{estimate} sub-phase. Therefore at any time algorithm has to keep four variables which are the boundaries $a$ and $b$, proposed estimation $u$, and a counter to estimate rank of $u$. For this algorithm data size $n$ should be given in order to decide how to divide stream into pieces. By adding one more variable, one can remove this requirement of knowing $n$ beforehand. This extra variable is used to remember the current iteration number, and stream is chopped into sub-streams with exponentially increasing length on iteration number. Each iteration instantiates a $Selection$ algorithm with current sub-stream length. The proved accuracy guarantee can be achieved when the overall stream is very large. In experiments, to overcome this requirement on every large streams we set $\delta=0.99$, and the version without knowing $n$ in advance is evaluated to make comparisons. \footnote{McGregor and Valiant \cite{Mcgregor12shiftingsands} gave a new algorithm using the same space, proving improved approximation with accuracy $n^{1/3+o(1)}$ can be achieved. 
This algorithm is more complicated to implement and also has qualitatively similar behaviour as the algorithm we have implemented here.}

\section{Empirical Evaluations}
\label{sec:evaluations}

\begin{figure*}[t!]
\centering
	$
	\begin{array}{cc}
	\epsfig{file=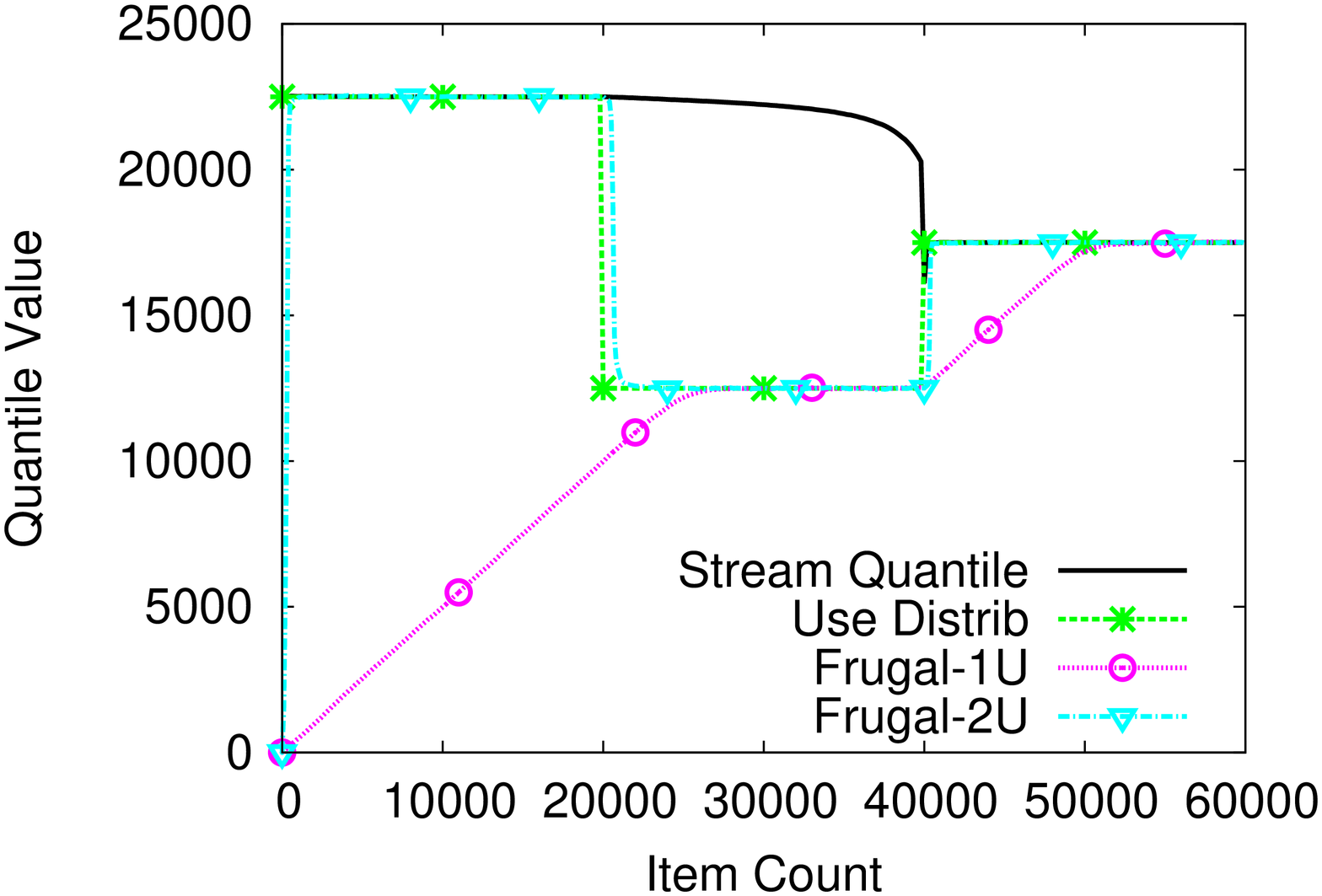, width=\figurewidthJ, angle=-90} &
	\epsfig{file=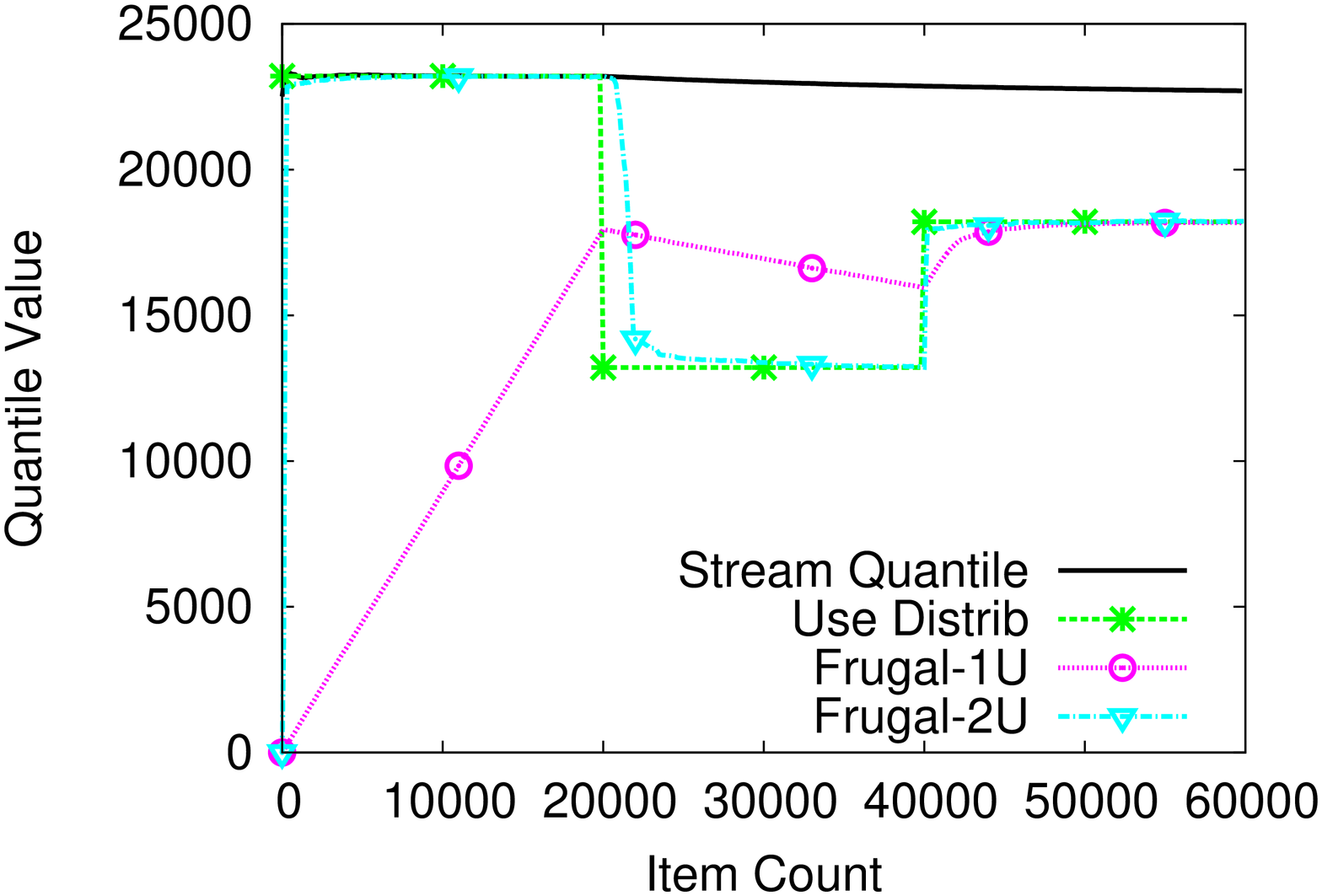, width=\figurewidthJ, angle=-90} \\
	\mbox{\bf (a)} & \mbox{\bf (b)}
	\end{array}
	$
	\caption{Evaluation on one stream generated from three Cauchy distributions. (a) Median estimation. (b) 90-\% quantile estimation. The change of $Use$-$Distrib$ curve indicates the change of underlying distribution. \twoUany algorithm converges to new distribution quantiles significantly faster than \oneUany.}
\label{plot:changing_cauchy}
\end{figure*}

In this section we evaluate our algorithms on both synthetic and real world data. For synthetic data we consider two scenarios, one when data arrive from a static distribution, and another when the distribution changes mid-stream. These tests allow us to demonstrate that our algorithms perform well on estimating stream quantiles for both scenarios. For real world data we evaluate on data from HTTP streams \cite{Bissias05httpstreams} and twitter user tweets data, where our goals are to evaluate median and \ninty quantile estimates of TCP-flow duration and size, and twitter users' tweet intervals. As mentioned earlier the structure of our algorithms allow us to estimate quantiles for every remote website or user with minimum (1-2 in-memory variables per data stream) memory requirement, and which quantile to estimate can be shared by all streams.

Instead of evaluating the absolute error of quantile estimation, we evaluate how far the estimate is from the true quantiles, the relative mass error. For example if the estimate
of 90-\% quantile turned out to be 89-\% quantile the error is then 0.01.

From Section~\ref{subsec:approach}, we know the initial estimations of our algorithms only affect the number of steps needed to approach true quantile, but not their stability in long run. Throughout our experiments, we initialize \oneUany and \twoUany algorithms estimates with $0$ (in practice we can also initialize them with the first stream item). For non-constant memory algorithms \gk and \qdigest, when we limit the memory budget to a small amount (\eg~20 units of their in memory data structure) they don't achieve accurate quantile estimations and perform worse than our \oneUany and \twoUany, but when given sufficient size of memory (\eg~500 units) they can perform well.

\subsection{\textbf{Synthetic Data}}
In this section we evaluate algorithms on data streams from a Cauchy distribution (density function $f(x) = \frac{\gamma}{\pi(\gamma^2 + (x-x_0)^{2}}$). 
The reason we picked Cauchy is because it has a high probability of outliers, indeed the expected value of a Cauchy random variable is infinity, and thus we can demonstrate that our algorithms work well in the presence of outliers.

\para{Static distribution.} For our experiments we fixed $x_0=10000$ and $\gamma=1250$. We draw $3 \times 10^4$ samples and explore estimation convergence. We let \oneUany algorithm start from 0, and quite as expected it took a long journey to approach the true quantile\footnote{Note, this is an inherent property of our algorithm, because the step is fixed at 1, if the range and/or acceptable error are known in advance the convergence can be improved. Our 2-variable algorithm does not need such knowledge}. \twoUany algorithm also starts its estimation from 0, but with dynamic $step$ size throughout the updates.

The curve $Stream~quantile$ in Figure \ref{plot:static_median} (and in other figures throughout our evaluation) shows the cumulative quantiles of a stream. Not only for \oneUany and \twoUany, but we see that each algorithm needs some time (some amount of stream items) before getting to a stabler quantile estimation. When memory is insufficient for the non-constant memory algorithms, estimation performance degrades much. Due to smaller fixed update size of \oneUany, it takes much longer travel than \twoUany to reach stream quantiles. 

\begin{figure*}[ht!]
  \centering
  \subfigure[{}]
  {\psfig{figure=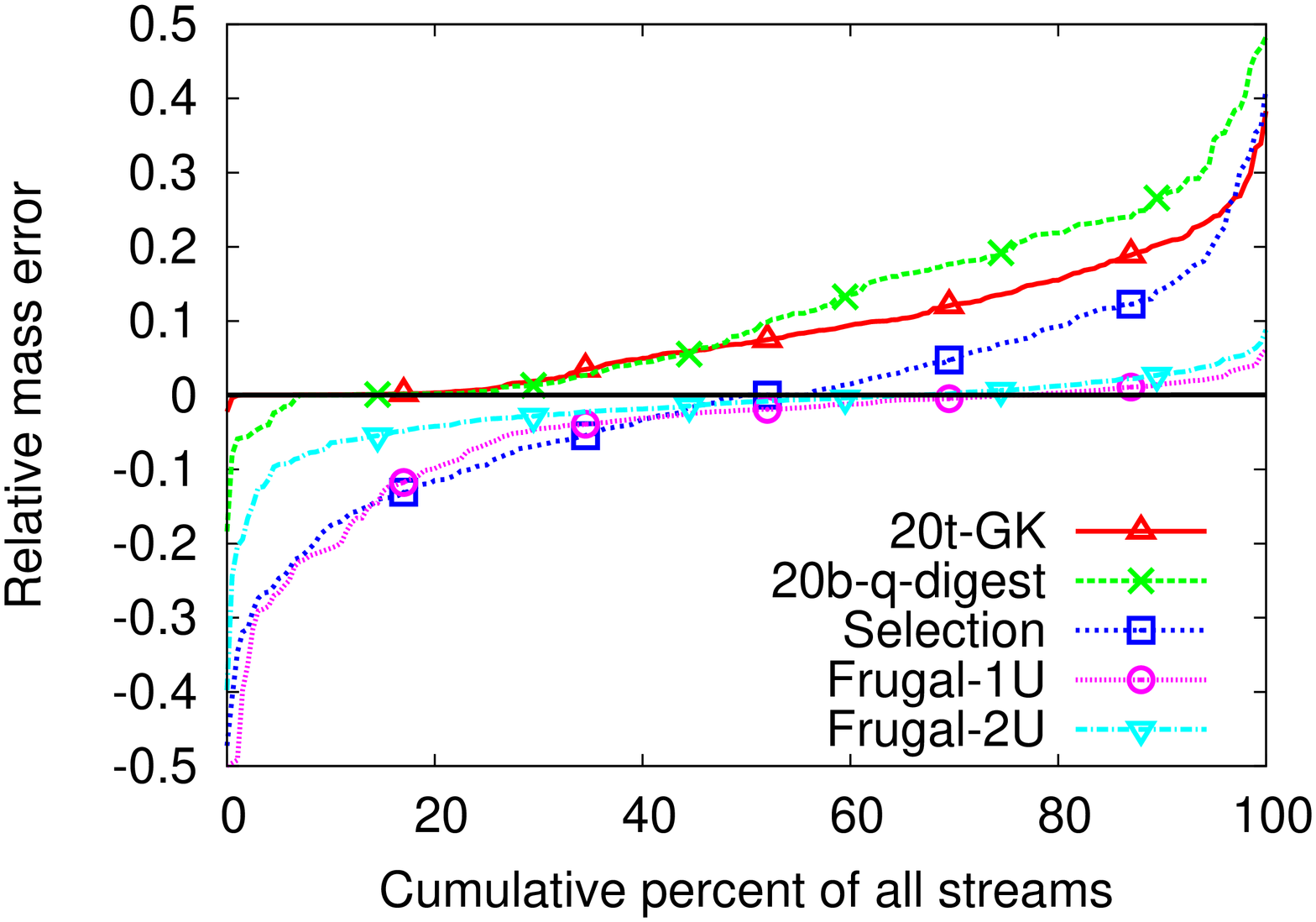, width=\figurewidthJ, angle=-90}}
  \subfigure[{}]
  {\psfig{figure=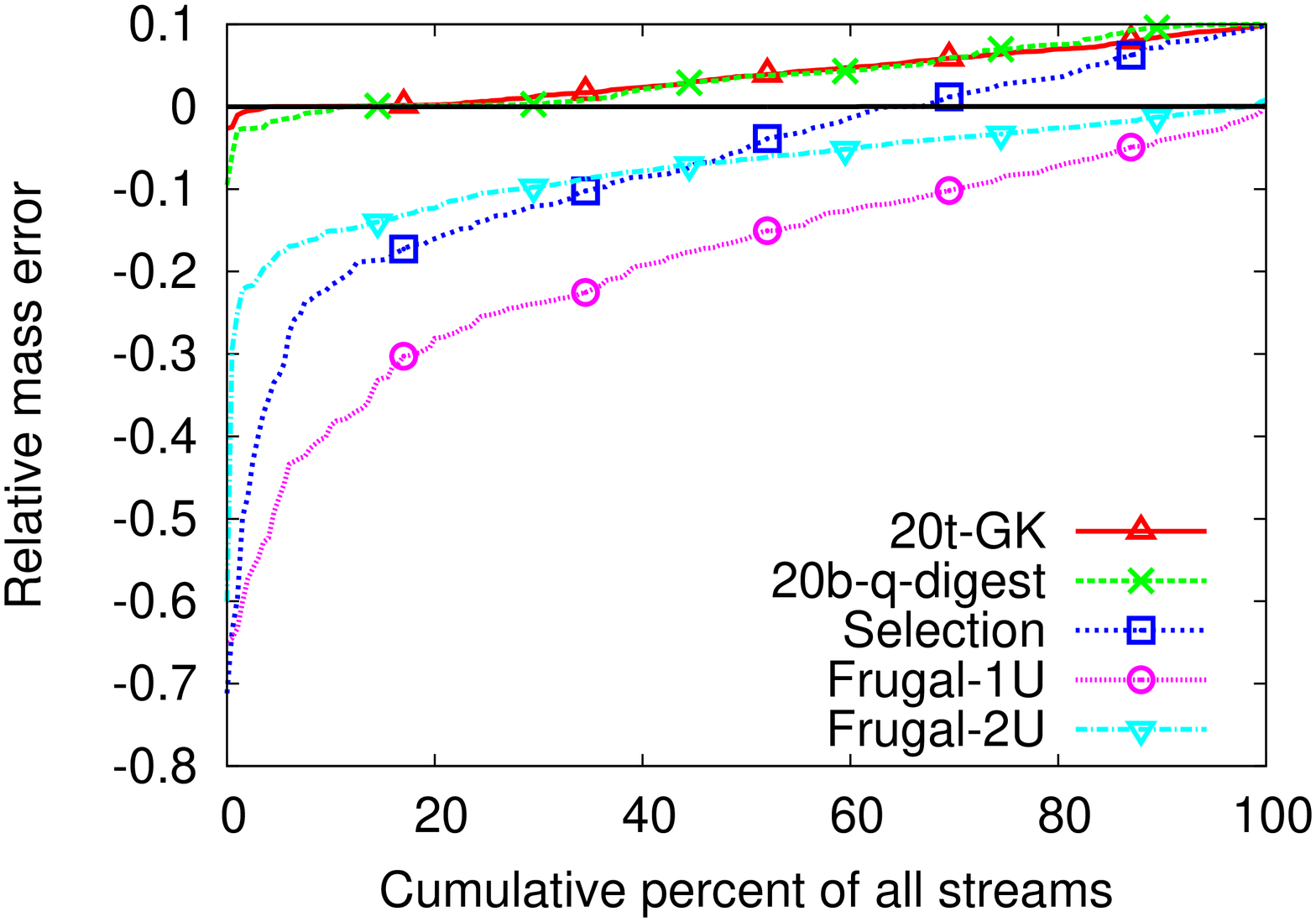, width=\figurewidthJ, angle=-90}}
  \caption{Evaluation on 419 TCP-flow size streams. (a) median estimation. (b) 90-\% quantile estimation}
  \label{plot:allsitesflowsize}
\end{figure*}

\begin{figure*}[ht!]
  \centering
  \subfigure[{}]
  {\psfig{figure=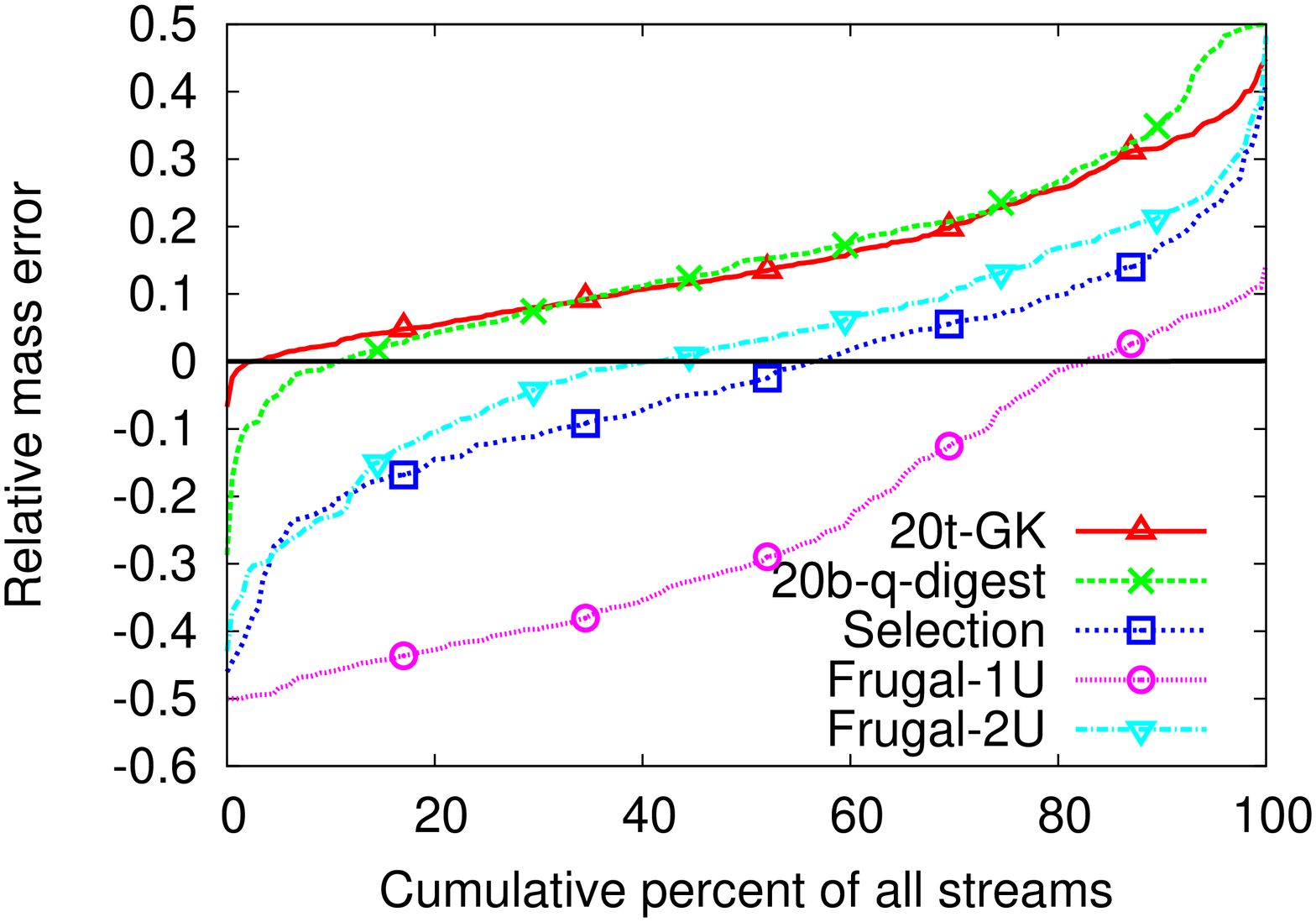, width=\figurewidthJ, angle=-90}}
  \subfigure[{}]
  {\psfig{figure=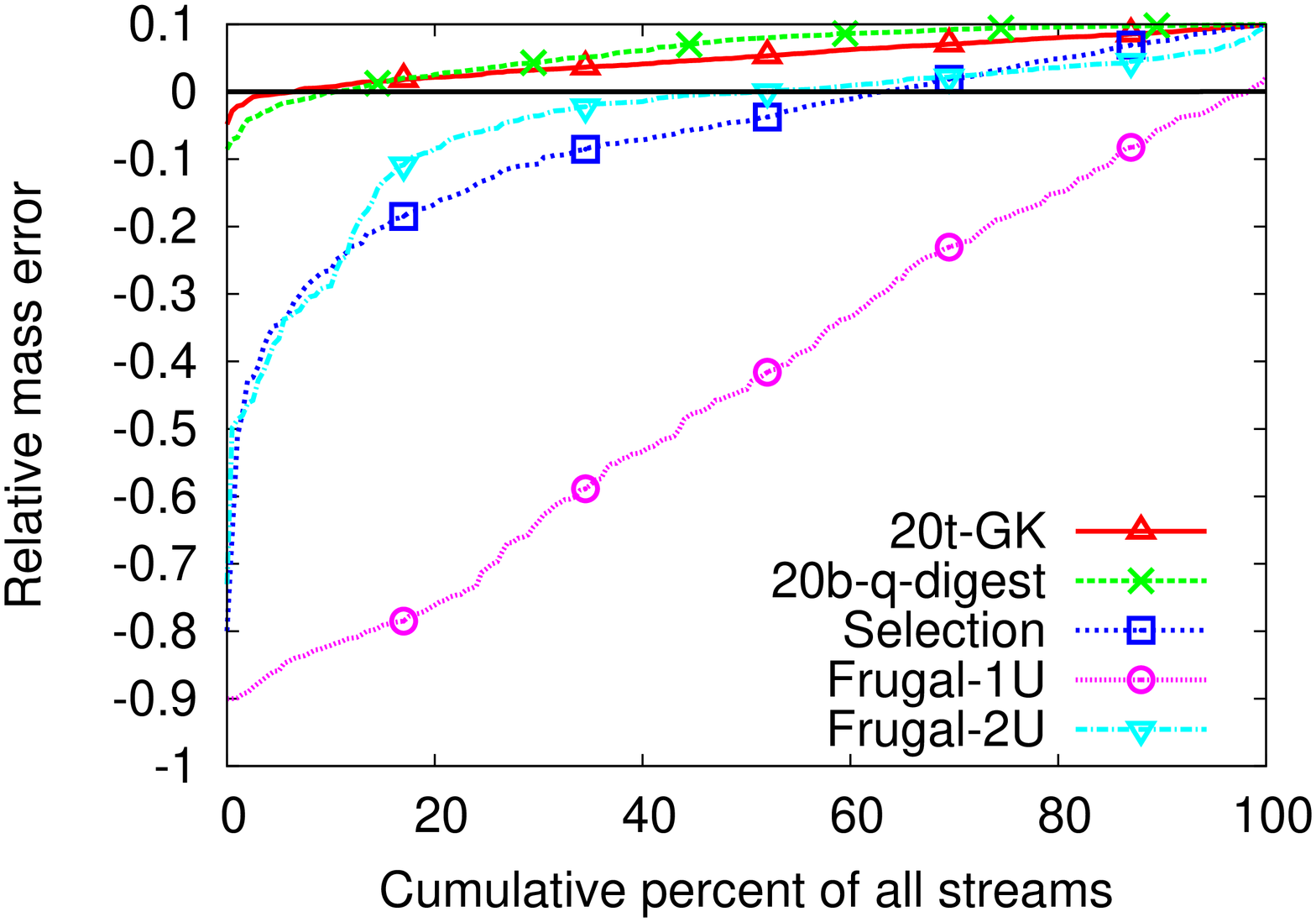, width=\figurewidthJ, angle=-90}}
  \caption{Evaluation on 419 TCP-flow duration streams. (a) median estimation. (b) 90-\% quantile estimation}
  \label{plot:allsitesduration}
\end{figure*}

\para{\textbf{Dynamic distribution.}}
Since other algorithms in comparison are not built for estimating changing distributions, we only evaluate \oneUany and \twoUany in the scenario where the underlying distribution of stream changes. We generate three sub-streams drawn from three different Cauchy distributions and feed them one by one to our algorithms to estimate stream quantiles. For each of the three sub-streams we sample $2\times10^4$ items in value domains [10000, 15000], [15000, 20000] and [20000, 25000] respectively.

Figure \ref{plot:changing_cauchy} shows the median and \ninty quantile estimations only for \oneUany and \twoUany algorithms. Those sub-streams are ordered by their medians in the order of highest, lowest and middle, then they are feed to algorithms one by one. 
For other algorithms they either need to know the value domain as input or they try to learn upper and lower bounds for the quantile in query, therefore if the stream underlying distribution changes their knowledge about stream are out-dated hence quantile approximations are probably not accurate. $Stream$-$quantile$ curve shows the cumulative stream quantiles, and this is the curve which those algorithms try to approximate if the combined stream is of interest at the beginning. But in this figure we want to show that our \oneUany and \twoUany are doing a different job. $Use$-$Distrib$ curve shows the quantile values for each sub-distribution. The change of $Use$-$Distrib$ curve indicates the change of underlying distribution. We can see that our algorithms are trying to reach new distribution's quantile when the stream underlying distribution changes. It is only that \oneUany takes longer time to approach new distribution's quantiles, while \twoUany can make ``sharper'' turns in its quantile estimations when distribution changes. \oneUany in Figure \ref{plot:changing_cauchy}.(b) leaves a steeper approaching trace to \ninty quantile than estimating median in Figure \ref{plot:changing_cauchy}.(a), because it is more biased to move estimate towards one direction (getting larger). 

One counter argument is that the property of adapting to changing distribution's new quantile might be a simultaneous disadvantage, because it makes the algorithms vulnerable to short bursts of "noise". However since the adjustment taken by \oneUany is 1, when stream domain is large the shifting from true stream quantile caused by short bursts will not affect much in terms of relative mass error. For \twoUany it is true that $step$'s increment and decrement function $f$ should be picked to trade-off between convergence speed and stability when bursts or periodic patterns are apparent in streams. But once after reaching a close estimate of true quantile, the decreasing $step$ value is able to buffer the impact of some value bursts.

\begin{figure*}[ht!]
  \centering
  \subfigure[{}]
  {\psfig{figure=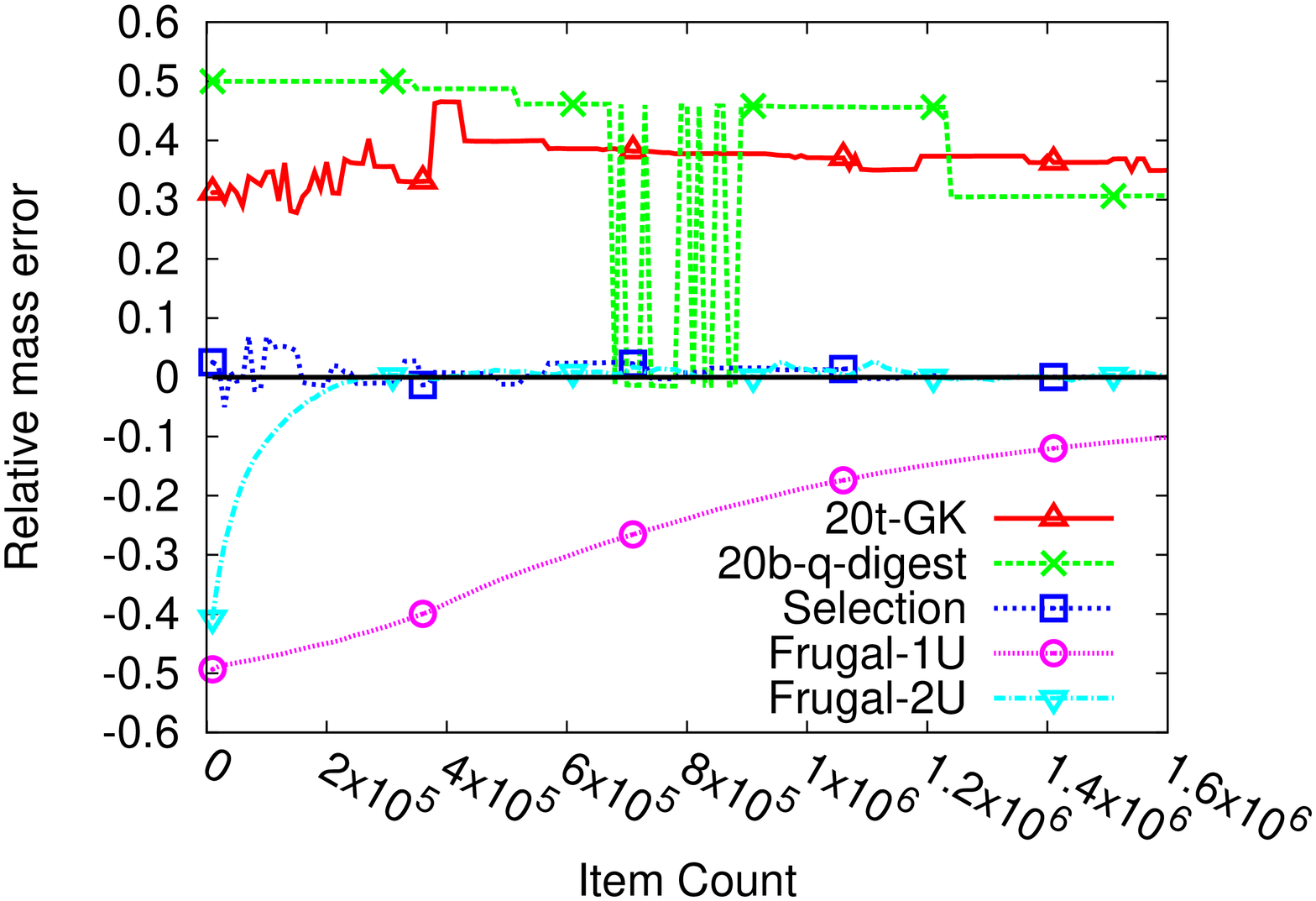, width=\figurewidthJ, angle=-90}}
  \subfigure[{}]
  {\psfig{figure=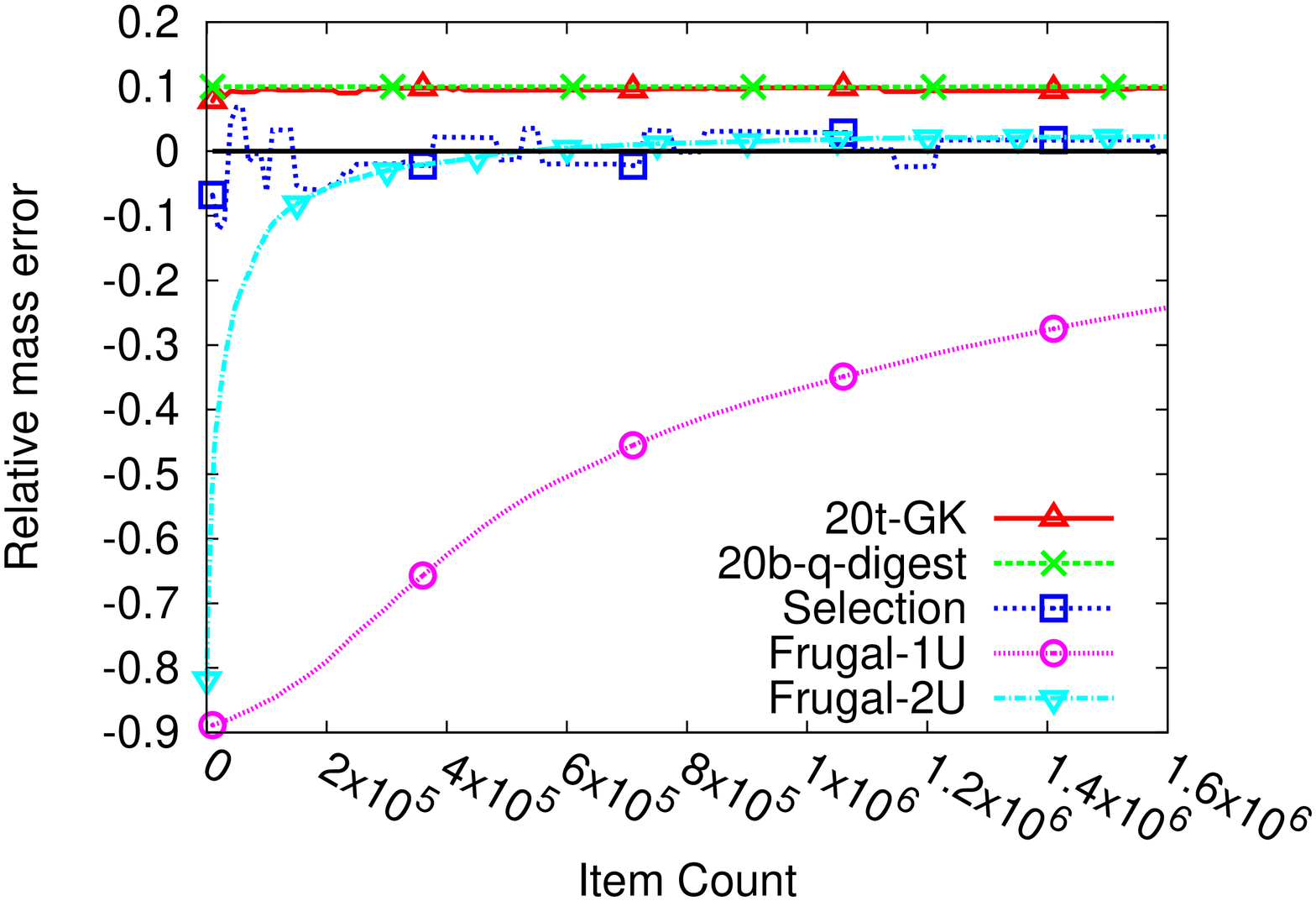, width=\figurewidthJ, angle=-90}}
  \caption{Evaluation on TCP-flow duration stream of month 2004-03. (a) median estimation. (b) 90-\% quantile estimation}
  \label{plot:combinedduration}
  \vspace{-0.15in}
\end{figure*}

\begin{figure*}[ht!]
  \centering
  \subfigure[{}]
  {\psfig{figure=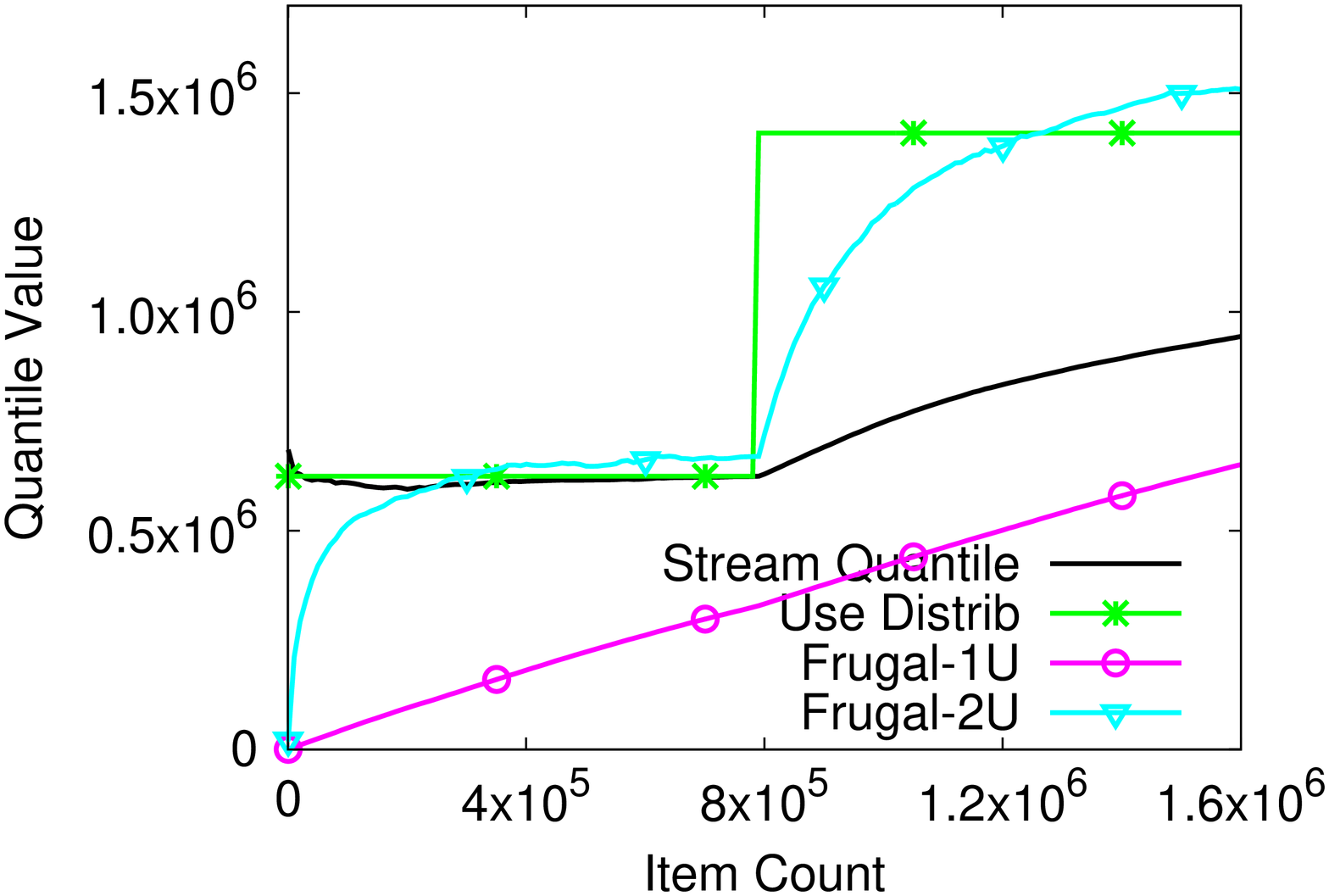, width=\figurewidthJ, angle=-90}}
  \subfigure[{}]
  {\psfig{figure=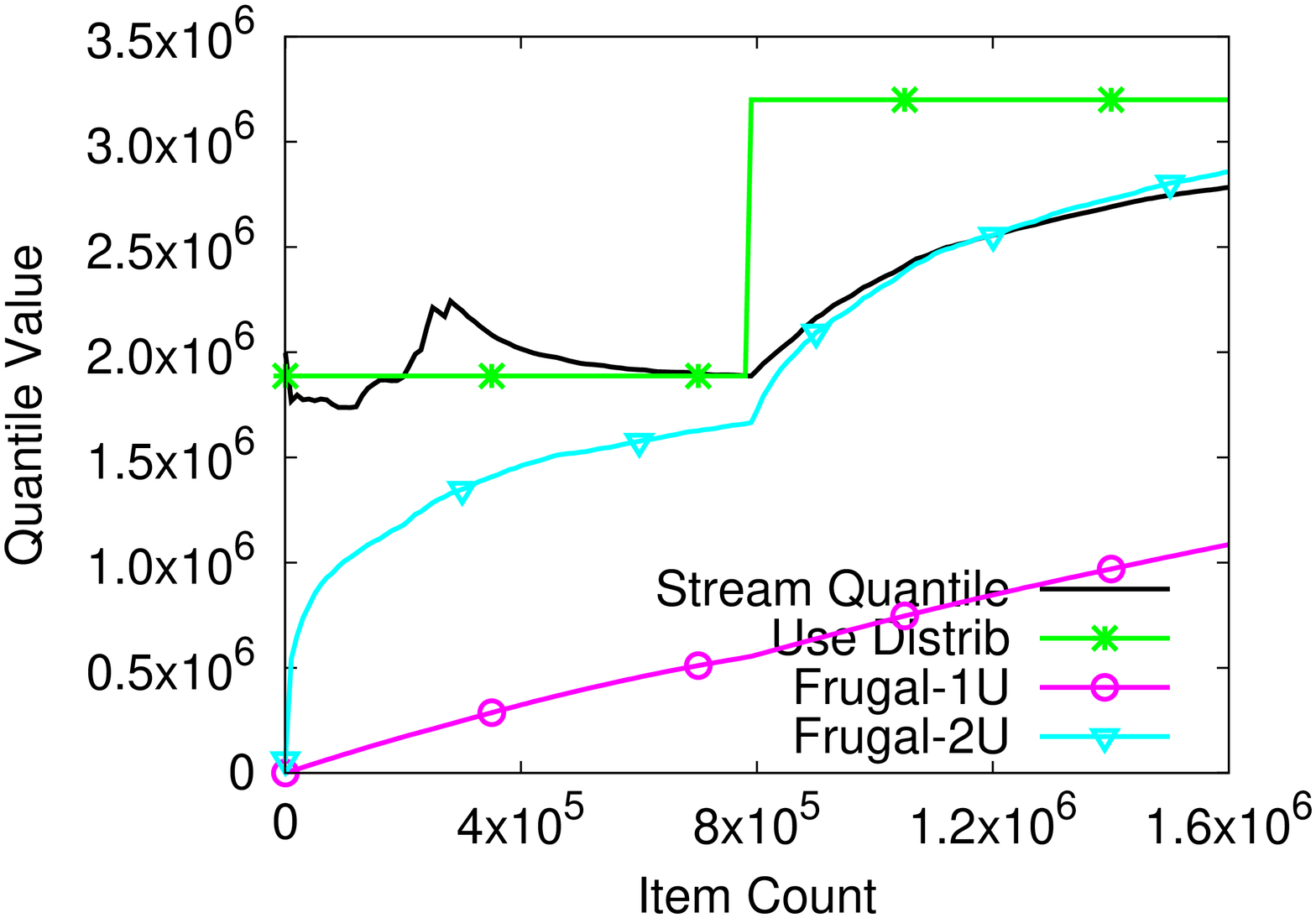, width=\figurewidthJ, angle=-90}}
  \vspace{-0.05in}
  \caption{Evaluation on TCP-flow duration stream of month 2003-12, with dynamic distribution. (a) median estimation. (b) 90-\% quantile estimation}
  \label{plot:combinedchanging}
  \vspace{-0.15in}
\end{figure*}

\subsection{\textbf{TCP-flow Data}}
From an HTTP request and response trace \cite{Bissias05httpstreams} collected for over a period of 6 months, spanning 2003-10 to 2004-03, we extract out TCP-flow durations (in millisecond\footnote{If use microsecond, the quantile values are too large for evaluation, where 90\% of the stream medians are above 260,057, but more than 80\% of the stream sizes are less than 20,000. Then \oneUany and \twoUany do not have much chance to get close to stream quantiles.}) and sizes (in bytes) between local clients and 100 remote sites, and order them by connections set up times to form streams. In this experiment we first evaluate on streams generated with each of those 100 sites in each of the 6 months. Therefore in total we have 600 streams. But in final performance evaluations we filter out streams with length less than 2000 items and end up with 419 used streams. Finally we collect the last estimations for median and 90-\% quantile by all algorithms. 

Figure~\ref{plot:allsitesflowsize} shows the relative mass error and cumulative percent of all 419 streams on estimating median and 90-\% quantile of flow size streams. We can see that in estimating median and 90-\% quantile for TCP-flow size streams, Figure \ref{plot:allsitesflowsize}.(a), \oneUany and \twoUany perform better than or comparable with other algorithms, with more than 90 percent of the last median estimations in error range [-0.1, 0.1]. In comparison, $t=20$  for \gk and $b=20$ for \qdigest are not enough to arrive at close estimations, and $Selection$ algorithm needs much longer streams. Note that in relative mass error figures, the overestimate errors are bounded by 0.5 and 0.1 respectively for median and \ninty quantile estimations. In \ref{plot:allsitesflowsize}.(b) \oneUany under-estimates \ninty quantile for a large portion of the streams due to insufficient stream sizes and relatively larger \ninty quantile values (90\% of the stream \ninty quantiles are larger than 4,354 while more than half of the stream sizes are less than 8,500). Although \twoUany makes under-estimates most of the time for \ninty quantile, in terms of estimation error range its performance does not degrade much.

 Figure~\ref{plot:allsitesduration} shows the performance comparison on 419 TCP-flow duration streams. In estimating medians of TCP-flow duration streams, Figure \ref{plot:allsitesduration}.(a), \oneUany and \twoUany perform worse than working on flow size streams. After examining the data, we found that in duration streams periodic patterns are apparent, where a series of large duration values are followed by a series of much smaller duration values. These patterns add noise to \oneUany and \twoUany, but still \twoUany performs better than \gk and \qdigest which use more than 10 times in memory variables.

In the situations where there are millions of streams to be processed simultaneously, statistical quantities about more general groups can help understand the characteristics of different groups. In HTTP request and response trace, streams generated by remote site can also be considered as \textbf{GROUPBY} application to understand the communication patterns from local clients to different remote sites. 

\begin{figure*}[ht!]
  \centering
  \subfigure[{}]
  {\psfig{file=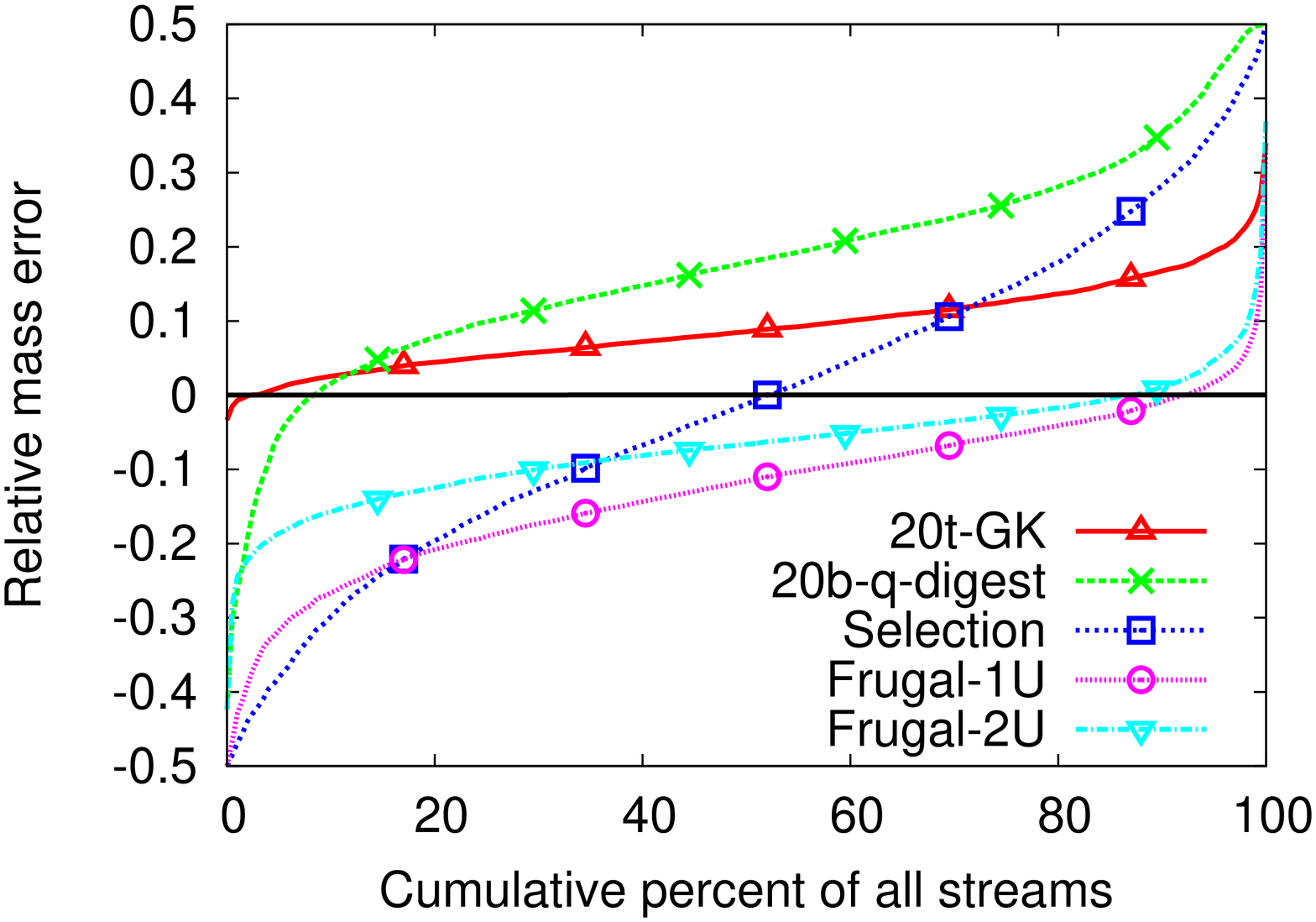, width=\figurewidthJ, angle=-90}}
  \subfigure[{}]
  {\psfig{file=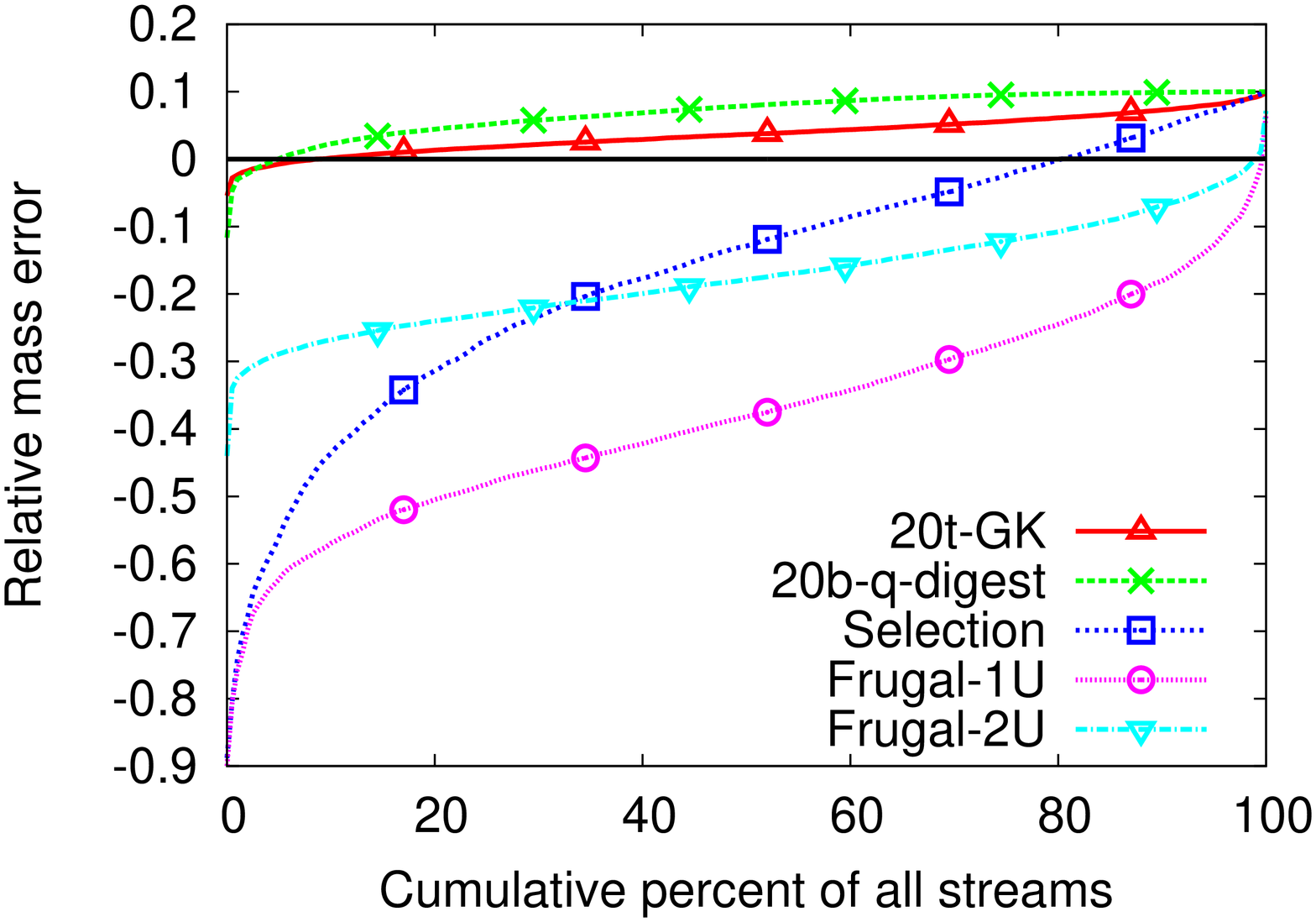, width=\figurewidthJ, angle=-90}}
\caption{Evaluation on 4414 twitterers' tweet interval streams. (a) median estimation. (b) 90-\% quantile estimation.}
\label{plot:twitter-all}
\vspace{-0.1in}
\end{figure*}

\begin{figure*}[ht!]
  \centering
  \subfigure[{}]
  {\psfig{file=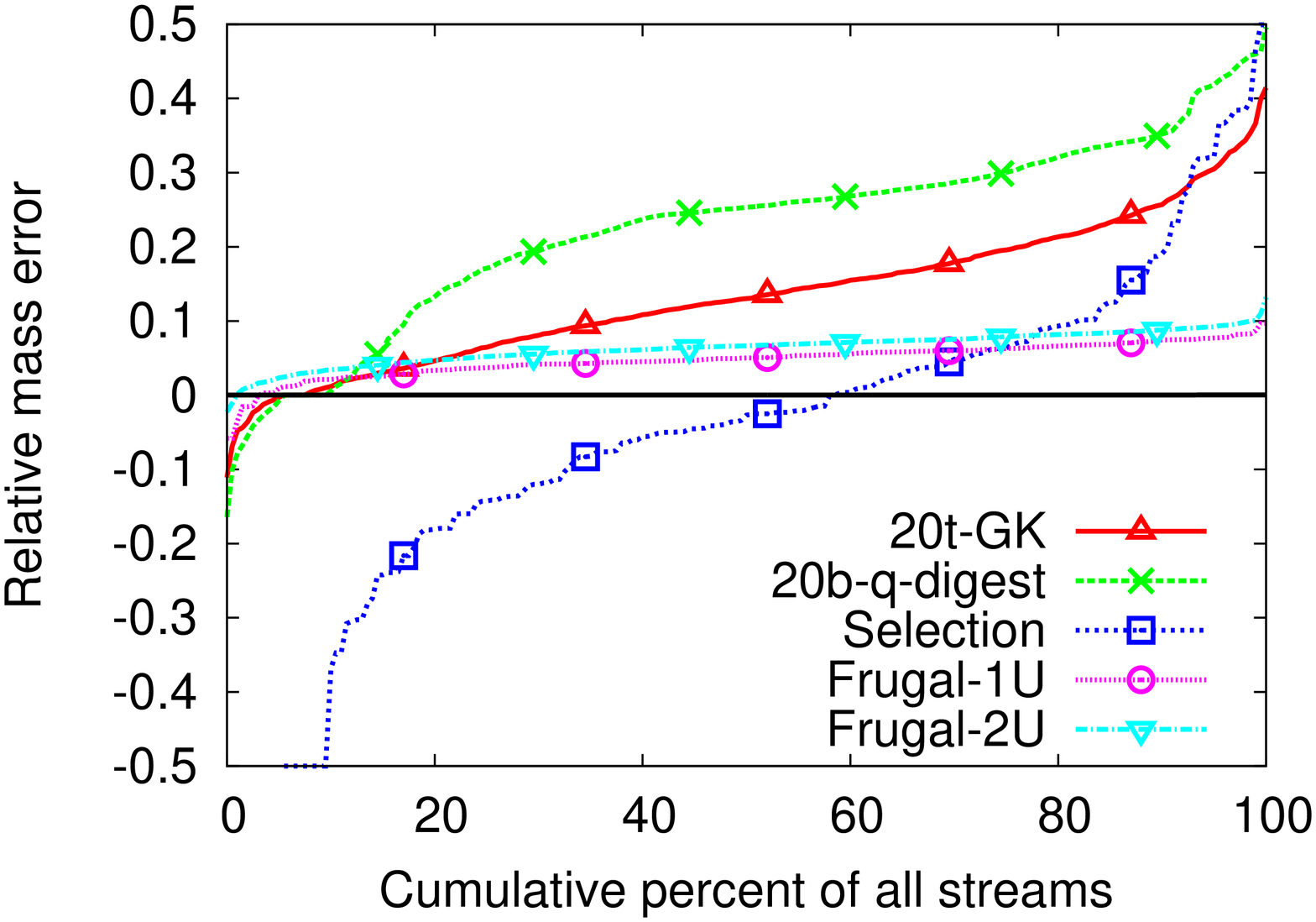, width=\figurewidthJ, angle=-90}}
  \subfigure[{}]
  {\psfig{file=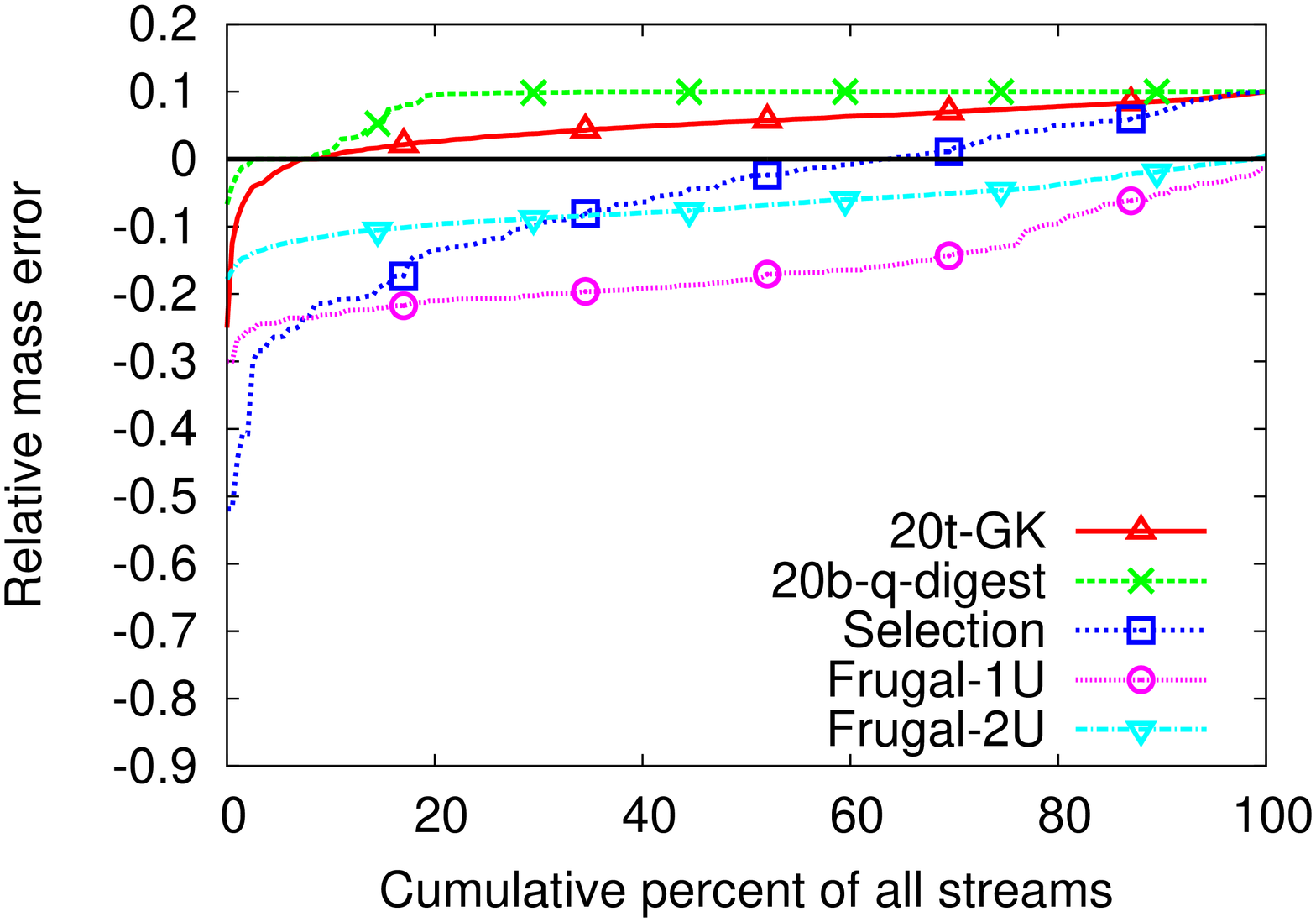, width=\figurewidthJ, angle=-90}}
\caption{Evaluation on 905 daily tweet interval streams. (a) Median estimation. (b) 90-\% quantile estimation}
\label{plot:alldays}
\vspace{-0.1in}
\end{figure*}

Note that stream size should be large for \oneUany and $Selection$ algorithms to settle at estimations close to true quantiles. We evaluated all algorithms on another \textbf{GROUPBY} application on this HTTP trace data, where connections with all 100 sites in each month are combined by their creation time. This simulates the viewpoint from trace collecting host. Algorithms are evaluated on each month's combined streams. For brevity here we present the results from evaluation on combined streams of month 2004-03, which contains one of the largest by month stream, and the results are similar for other months (except the distribution changing stream we will see later). This combined stream has about $1.6\times 10^{6}$ items. Figure~\ref{plot:combinedduration} 
presents the results on estimating median and \ninty quantile of TCP-flow duration
stream. This duration stream's items are in unit of microsecond, because we have a large enough stream for algorithms to approximate large quantiles, and observe how algorithm estimations approach true quantiles. In this stream we have median and \ninty quantile values at about 544,267 and 1,464,793 respectively. Due to these large quantile values \oneUany shows a slower convergence to true stream quantile, while \twoUany handles this problem much better. $Selection$ converges to [-0.1, 0.1] relative mass error region after about $2\times 10^5$ items, but it is oscillatory thereafter and needs much more items to stabilize. In contrast, although \oneUany and \twoUany need relatively more stream items to reach a large true quantile their estimations are relatively stabler. In Figure~\ref{plot:combinedduration}.(a), $b=20$ \qdigest gives very oscillatory median estimation around $8\times 10^5$, and from the curve it seems converging to stream median but apparently it needs much more stream items.

\textbf{Dynamic distribution.} The TCP-flow duration stream of 2003-12 changes its distribution in the middle due to the change of contributing set of remote sites. Therefore it serves well for the purpose of evaluating on stream with dynamic distribution. This stream length is about $1.6\times 10^6$, and durations are in unit of microsecond. Since other algorithms are not designed for dynamic distribution streams, we hide them from Figure  \ref{plot:combinedchanging}. $Stream$-$quantile$ shows the cumulative stream median and \ninty quantile values, and $Use$-$Distrib$ gives us the median and \ninty quantile values of each distribution. In Figure \ref{plot:combinedchanging}.(a) and (b), we show how quantile values (y-axis) change over time against \oneUany and \twoUany estimations. The stream median and \ninty quantile change about mid-stream, \twoUany can reach close median estimate in Figure \ref{plot:combinedchanging}.(a) before distribution change. Then it takes a clear "turn" to approach new distribution's median in the second half. Although at the end of this stream \twoUany estimation is larger than second distribution's true median (due to the large $step$ value cumulated while adapting to new distribution), we can see it shows the trend to stop increasing and converge to true median. And we expect its estimation to fall back to true median as stream continues. \twoUany shows similar behaviour in estimating \ninty quantile in Figure \ref{plot:combinedchanging}.(b),  but due to larger quantile value, it does not get the chance to reach close estimation before stream changes or ends. On the other hand, \oneUany takes much more items to reach stream quantile values, so in both plots it just leaves an almost linear trace to chase stream quantiles.

\subsection{\textbf{Twitter Data Set}}
From an on-line twitter user directory\omt{~\cite{website:twitterdirectory}}, we collected 4554 users over 80 directories (e.g. Food and Business). Those tweets from individual users form 4554 sub-streams in the ocean of all tweets. 
We extracted the intervals (in seconds) between two consecutive tweets for every user and then run our algorithms on those interval streams. This allows us to answer the question of ``what is the median inactive time for a given user across all?''.

Among the total 4554 twitterers, we removed the users with less than 2000 tweets since we need a decent number of data items to reflect the true distribution and allow our algorithms to reach true quantiles. Since twitter does not store more than 3200 tweets of a single user, therefore at the time of data collection the maximum length of a single user's interval stream is 3200. So finally we evaluated our algorithms on 4414 twitter user interval streams, and collected the last estimations for median and \ninty quantile.

Figure~\ref{plot:twitter-all} shows the relative mass error and cumulative percent of all 4414 interval streams. In Figure~\ref{plot:twitter-all}.(a) we see that about 70 percent of the last median estimation by \oneUany are under-estimating (less than -0.1). Because we initiated quantile estimations from 0, however interval stream median (and \ninty quantile) values can easily be tens of thousands (about 90\% of interval streams have \ninty quantiles larger than $10^4$), within 2000 steps it can not fully reach true medians. \twoUany performs much better than \oneUany algorithm, with more than 80 percent of the last median estimations in error range [-0.1, 0.1]. Figure~\ref{plot:twitter-all}.(b) shows that when estimating \ninty quantile, which are much larger values, as expected \oneUany cannot reach true quantile when the stream items are few (94\% of twitter user interval streams have \ninty quantiles larger than 3,200, while only about 6\% of theirs streams have size 3,200). Again \twoUany shows its advantages over \oneUany but it also needs longer streams to reach true quantiles. In comparison, $t=20$  for \gk and $b=20$ for \qdigest are not affected by stream sizes, however \selection algorithm needs much longer streams. Again note that from this figure, the overestimate errors are bounded by 0.5 and 0.1 respectively for median and \ninty quantile estimations, because relative mass error is measured.

For a database there are various meaningful group by applications, such as group by geo-location and age for an on-line social network database. To simulate such \textbf{GROUPBY} application, we evaluate our algorithms on the combined tweet interval streams on each day. We merge tweet interval streams from all 4554 twitterers in our dataset, and sort all the intervals based on the time they were created. We divide the combined interval stream into segments by day, and in total our tweet interval data spanning 1328 days from 2008 to 2011. We ran our algorithms on each day's data and take the last estimations from algorithms to evaluate their accuracy. We filter out the days that have less than 2000 intervals in the daily stream, since small number of intervals in the stream doesn't give enough chance for our algorithms to approach true quantiles. After filtering process, we have 905 days left. Figure~\ref{plot:alldays} shows the cumulative percent of all days against relative mass error, both median and \ninty quantile under-estimation problems in individual user interval streams are alleviated (in daily interval streams about 67\% of the streams have size larger than 3,200). Daily median estimation performance by \oneUany in Figure~\ref{plot:alldays}.(a) demonstrate that it can reach close estimation before the daily interval streams end. In Figure~\ref{plot:alldays}.(b), for \ninty quantile on most of the days \oneUany algorithm underestimates the true quantiles by using update size of 1. For \twoUany, for both median and \ninty quantile estimations almost all last estimations are in error range [-0.1, 0.1]. Again in comparison, $t=20$  for \gk and $b=20$ for \qdigest are not enough to get close estimations, and \selection algorithm needs much more stream items.

Throughout our extensive experiments on synthetic and real-world data, for stochastic streams given enough number of data items in the stream, our 1 and 2 variables stochastic algorithms can achieve quite comparative accuracy against other non-constant and constant memory algorithms, while using much less memory and being very efficient for per item update.

\section{Conclusions and Future Directions}
\label{sec:conclusion}
We have introduced the concept of frugal streaming and  presented algorithms that can estimate arbitrary quantiles using $1$ or $2$ unit memories. This is very useful when we need to estimate 
quantiles for each of many groups, as applications demand in reality. These algorithms do not perform well with adversarial streams, but we have mathematically analyzed the $1$ unit memory
algorithm and shown fast approach and stability properties for 
stochastic streams. Our analysis is non-trivial, and we believe it provides a framework for 
analysis of other statistical estimates with stochastic streams. Further we have reported extensive
experiments with our algorithms and several prior quantile algorithms on synthetic data as well as 
real dataset from HTTP trace and Twitter.

To the best of our knowledge our algorithms are the first that perform well with $2$ or less persistent variables per group. In contrast, other regular streaming algorithms,  while having other desirable properties, perform poorly when pushed to the extreme on memory consumption like we do with our
frugal streaming algorithms. 

Our work has initiated frugal streaming, but much remains to be done. First, we need mathematical
analyses of $2$ or more memory algorithms and at this moment, it looks quite non-trivial. We also need frugal streaming algorithms for other problems such as distinct count estimation and others, that are critical for streaming applications. Finally, as our experiments and insights indicate, frugal
streaming algorithms work with so little memory of the past that they are adaptable to changes in the stream characteristics. It will be of great interest to understand this phenomenon better.

\section*{Acknowledgements} 
This work was sponsored by the NSF Grant 1161151: AF: Sparse Approximation: Theory and Extensions.

\bibliographystyle{abbrv}

\bibliography{lib}

\end{document}